\documentclass[sigconf,nonacm]{acmart}
\usepackage{booktabs}
\usepackage{amsmath}%
\usepackage{xcolor}
\usepackage{colortbl}
\usepackage{tabularx}
\usepackage{tikz}
\usetikzlibrary{arrows.meta,positioning,calc,fit,shapes.misc,shapes.geometric}
\usepackage[linesnumbered,ruled,vlined]{algorithm2e}
\newcommand{\Units}{\mathcal{U}}
\newcommand{\Ctx}{C}
\newcommand{\Req}{\mathrm{Req}}

\newcommand{\Exc}{\mathrm{Exc}}
\newcommand{\Rq}{\mathcal{R}(q)}
\newcommand{\Rules}{\mathcal{E}}
\newcommand{\Atoms}{\mathcal{V}}
\newcommand{\Cl}{\mathrm{Cl}}
\newcommand{\facts}{\mathrm{facts}}
\newcommand{\rules}{\mathrm{rules}}
\newcommand{\suff}{\mathrm{suff}}
\newcommand{\Target}{Q}
\newcommand{\Proj}{\Pi}
\newcommand{\Sec}{\mathrm{S}_2}
\newcommand{\method}{\textsc{Certify}}
\newcommand{\ours}{\rowcolor{black!8}}%

\newcommand{\nAWBPooledSrBigLabDDec}{+2.2}

\newcommand{\nAWSPooledSrLabDDec}{+0.7}

\newcommand{\nAbMemPooledPoolD}{+11.4}

\newcommand{\nAbMemPooledRagD}{+16.8}

\newcommand{\nAbPooledClosedDec}{35.3}

\newcommand{\nAbPooledPoolD}{+11.5}

\newcommand{\nAsDSPooledCert}{89.5}

\newcommand{\nAsDSPooledPool}{73.5}

\newcommand{\nAsPooledCert}{80.0}

\newcommand{\nAsPooledPool}{56.7}

\newcommand{\nAsRPooledCert}{82.9}

\newcommand{\nAsRPooledPool}{69.8}

\newcommand{\nAudMU}{4.6}

\newcommand{\nAudTW}{2.6}

\newcommand{\nAudVAncTW}{96.7}

\newcommand{\nAudVUnTW}{20.0}

\newcommand{\nAudVerTWN}{88}

\newcommand{\nBTTWInfJudge}{+13.6}

\newcommand{\nBTTWInfPairs}{70}

\newcommand{\nBelowHalfTW}{nine}
\newcommand{\nBudHODec}{71.5}
\newcommand{\nBudHOLinks}{9.3}
\newcommand{\nBudHOLinksDef}{12.9}

\newcommand{\nBudMUDec}{60.4}
\newcommand{\nBudMULinks}{11.1}
\newcommand{\nBudMULinksDef}{20.0}

\newcommand{\nBudTWDec}{79.5}
\newcommand{\nBudTWLinks}{10.7}
\newcommand{\nBudTWLinksDef}{13.2}

\newcommand{\nCaHODropped}{7}

\newcommand{\nCaLab}{+4.6}
\newcommand{\nCaMUDropped}{26}

\newcommand{\nCaPubMax}{+42.9}

\newcommand{\nCaPubMin}{+12.0}
\newcommand{\nCaSrLab}{$-$0.8}
\newcommand{\nCaTWDropped}{8}

\newcommand{\nCauFireMU}{62.1}
\newcommand{\nCauFireTW}{88.0}

\newcommand{\nCauHoneMU}{52.8}
\newcommand{\nCauHoneTW}{95.1}

\newcommand{\nDjHOCert}{96.6}
\newcommand{\nDjMUCert}{96.0}
\newcommand{\nDjMax}{99.5}
\newcommand{\nDjMin}{90.3}
\newcommand{\nDjPubMin}{+12.0}
\newcommand{\nDjSrLabD}{$-$0.2}
\newcommand{\nDjSrSmD}{+4.9}
\newcommand{\nDjStwogGemD}{+3.3}
\newcommand{\nDjTWCert}{96.8}

\newcommand{\nDvOneShift}{$-$1.2}

\newcommand{\nDvTwoShift}{$-$1.7}
\newcommand{\nDvTwoSrLabD}{$-$2.7}

\newcommand{\nELBigDeepDec}{53.7}

\newcommand{\nELBigOneDec}{62.8}

\newcommand{\nFlHOAud}{23}

\newcommand{\nFlHONeither}{7}
\newcommand{\nFlHOOrig}{22}

\newcommand{\nFlHOTwin}{59}

\newcommand{\nFlMUAud}{17}

\newcommand{\nFlMUNeither}{3}
\newcommand{\nFlMUOrig}{31}

\newcommand{\nFlMUTwin}{52}

\newcommand{\nFlTWAud}{12}

\newcommand{\nFlTWNeither}{1}
\newcommand{\nFlTWOrig}{14}

\newcommand{\nFlTWTwin}{77}

\newcommand{\nHOPairs}{400}

\newcommand{\nHqRestMax}{42.0}
\newcommand{\nHqRestMin}{14.8}

\newcommand{\nHuAnc}{97.8}

\newcommand{\nHuDs}{89.6}
\newcommand{\nHuDsK}{0.74}
\newcommand{\nHuGem}{85.2}
\newcommand{\nHuGemK}{0.65}
\newcommand{\nHuGemNo}{16}
\newcommand{\nHuGemYes}{4}
\newcommand{\nHuLeakMU}{2.3}
\newcommand{\nHuLeakTW}{0.3}
\newcommand{\nHuLinks}{135}
\newcommand{\nHuOth}{50.0}

\newcommand{\nHuTwins}{100}
\newcommand{\nHuVer}{76.1}

\newcommand{\nHxCases}{13}

\newcommand{\nHxCertSAcc}{74.8}
\newcommand{\nHxCertTwin}{5.8}

\newcommand{\nHxMax}{42.0}
\newcommand{\nHxMin}{5.1}
\newcommand{\nHxN}{429}

\newcommand{\nHxShare}{29}
\newcommand{\nHxSrSAcc}{52.7}
\newcommand{\nHxSrTwin}{79.0}

\newcommand{\nHyPooledD}{+0.3}

\newcommand{\nIterStopMax}{99.5}
\newcommand{\nIterStopMin}{94.5}

\newcommand{\nLeakMU}{19.1}
\newcommand{\nLeakTW}{6.8}

\newcommand{\nLocChance}{48.2}
\newcommand{\nLocHit}{87.6}
\newcommand{\nLocShare}{59}

\newcommand{\nMUAnvAcc}{30.8}

\newcommand{\nMUAnvCalls}{22}

\newcommand{\nMUAnvSAcc}{74.5}

\newcommand{\nMUAnvTwin}{7.6}
\newcommand{\nMUAnvTwinGold}{2.1}

\newcommand{\nMUCertDec}{57.3}

\newcommand{\nMUCertTwin}{9.4}

\newcommand{\nMUIterAcc}{34.6}

\newcommand{\nMULenientAcc}{34.1}

\newcommand{\nMULenientTwin}{13.2}

\newcommand{\nMULinksAnvDec}{61.0}

\newcommand{\nMUOracleDec}{55.9}

\newcommand{\nMUSameDec}{49.2}

\newcommand{\nMUSrAcc}{47.2}

\newcommand{\nNdA}{74.0}
\newcommand{\nNdB}{59.3}
\newcommand{\nNdC}{21.8}

\newcommand{\nNdN}{661}

\newcommand{\nPrevGpMax}{83}
\newcommand{\nPrevGpMin}{47}

\newcommand{\nPrevVHOSr}{12}

\newcommand{\nPrevVMUSr}{19}

\newcommand{\nPrevVTWSr}{1}

\newcommand{\nPwMiss}{20.6}
\newcommand{\nPwMissDeep}{36.1}

\newcommand{\nPwRandMax}{20.5}
\newcommand{\nPwRandMin}{19.8}
\newcommand{\nQtwoD}{+3.0}

\newcommand{\nRbMUDec}{60.6}

\newcommand{\nRbTWDec}{79.2}

\newcommand{\nSaDSPooledDDec}{+0.5}

\newcommand{\nSaPooledDDec}{$-$0.8}

\newcommand{\nSaRPooledDDec}{$-$0.7}

\newcommand{\nSelPooled}{68.4}
\newcommand{\nSelWeb}{52.4}

\newcommand{\nSrMemHOD}{+31.7}
\newcommand{\nSrMemHOKnown}{28}
\newcommand{\nSrMemMUD}{+28.1}
\newcommand{\nSrMemMUKnown}{15}
\newcommand{\nSrMemTWD}{+40.0}
\newcommand{\nSrMemTWKnown}{53}

\newcommand{\nSrgRefPooledD}{+28.1}

\newcommand{\nSwClosN}{11}
\newcommand{\nSwClosOk}{10}
\newcommand{\nSwClosTwin}{9}

\newcommand{\nSwDisMax}{6.5}
\newcommand{\nSwDisMin}{0.5}

\newcommand{\nSwOrigOk}{3}
\newcommand{\nSwOrigRest}{27}

\newcommand{\nSwTwinAns}{14}

\newcommand{\nTWAnvAcc}{73.0}

\newcommand{\nTWAnvCalls}{17}

\newcommand{\nTWAnvSAcc}{85.8}

\newcommand{\nTWAnvTwin}{8.2}
\newcommand{\nTWAnvTwinGold}{1.6}

\newcommand{\nTWCertDec}{68.7}

\newcommand{\nTWCertTwin}{13.6}

\newcommand{\nTWLenientTwin}{20.6}

\newcommand{\nTWLinksAnvDec}{75.8}

\newcommand{\nTWOracleDec}{68.0}

\newcommand{\nTWSameDec}{52.9}

\newcommand{\nTWSrCalls}{4}

\newcommand{\nTableTwoPub}{eleven}

\newcommand{\nVDSPooledSrLabDDec}{+5.0}

\newcommand{\nVHOLinksAnvDDec}{+0.4}

\newcommand{\nVMULenientDDec}{+1.1}

\newcommand{\nVMULinksAnvDDec}{$-$0.6}

\newcommand{\nVMUPipeSmallDDec}{+6.8}

\newcommand{\nVMUStwogAllDDec}{+11.5}

\newcommand{\nVPooledHaltADDec}{+19.4}

\newcommand{\nVPooledLenientDDec}{+2.7}

\newcommand{\nVPooledPipeGemDDec}{+4.3}

\newcommand{\nVPooledSrLabDDec}{$-$1.1}

\newcommand{\nVPooledSrLabSmallDDec}{+4.2}

\newcommand{\nVPooledStwogDDec}{+11.4}
\newcommand{\nVPooledStwogGemDDec}{+3.0}

\newcommand{\nVRPooledSrBigLabDDec}{+0.6}

\newcommand{\nVTWLenientDDec}{+5.3}

\newcommand{\nVTWLinksAnvDDec}{+2.9}

\newcommand{\nVTWPipeSmallDDec}{+15.2}

\newcommand{\nVTWStwogAllDDec}{+24.1}

\newcommand{\nVTWStwogGemDDec}{+9.5}

\newcommand{\nWAudHO}{29.5}
\newcommand{\nWAudMU}{52.1}
\newcommand{\nWAudTW}{20.6}

\newcommand{\nWCHOAncDec}{60.7}

\newcommand{\nWCMUAncDec}{52.2}

\newcommand{\nWCPubMax}{10.8}
\newcommand{\nWCPubMin}{1.1}
\newcommand{\nWCTWAncDec}{63.9}

\newcommand{\nWCertSHODec}{49.8}
\newcommand{\nWCertSMUDec}{45.8}
\newcommand{\nWCertSTWDec}{50.4}

\newcommand{\nWGainBHO}{+6.6}
\newcommand{\nWGainBMU}{+4.7}
\newcommand{\nWGainBTW}{+11.9}
\newcommand{\nWGainSHO}{+8.8}
\newcommand{\nWGainSMU}{+4.2}
\newcommand{\nWGainSTW}{+13.3}

\newcommand{\nWHODeleted}{2.3}

\newcommand{\nWHOPairs}{251}

\newcommand{\nWHOTwins}{366}

\newcommand{\nWMUDeleted}{2.5}

\newcommand{\nWMUPairs}{526}

\newcommand{\nWTWDeleted}{1.9}

\newcommand{\nWTWPairs}{413}

\newcommand{\nWVMUPairs}{309}

\newcommand{\nWVTWPairs}{113}

\newcommand{\nWaGainMax}{8.6}
\newcommand{\nWaGainMin}{4.0}

\newcommand{\nWfRecall}{0.964}
\newcommand{\nWfRecallOrig}{0.965}

\begin{document}
\title{Knowing When the Evidence Is Enough: Requirement-Directed Search with Horn-Closure Certificates}
\author{Xinyu Zhang}
\affiliation{\institution{Fudan University}\country{}}
\email{zhangxinyu@fudan.edu.cn}
\author{Weiwei Sun}
\affiliation{\institution{Fudan University}\country{}}
\email{wwsun@fudan.edu.cn}
\begin{abstract}
After each Web search, a search-augmented model must decide whether its evidence suffices: answer, search again, or abstain.
Prevailing systems leave this to holistic model judgement, which can accept a context naming every entity yet lacking one required fact.
We instead make sufficiency structural: a per-query requirement hypergraph states which facts jointly yield another fact.
This raises two challenges: deriving requirements from the question alone, and linking each requirement to passages about the right entity.
\method{} therefore decomposes the question, searches per requirement, and accepts links only after entity checks.
Horn closure over the links yields a certificate, the minimal sufficient passages and their derivation, or an abstention naming open requirements.
Given correct retrieval, requirements and links, the decision is sound, and complete up to the entity check.
Pooled over three corpora with unanswerable twins, \method{} decides significantly more accurately than search agents, graph retrieval and sufficiency judges in default configurations.
On the same evidence, it matches a proprietary model's sufficiency label without an additional model call.
Over 21M Wikipedia passages, entity checks add 4 to 13 points of decision accuracy where gold evidence exists.
Finally, answers from its certificate are 80\% correct, against 57\% from all retrieved passages.
\end{abstract}
\maketitle

\section{Introduction}
\label{sec:intro}

\looseness=-1 Search-augmented language models~\cite{chen2025mindsearch,li2025searcho1} now answer a large share of information needs on the Web.
A Web search assistant or a search agent~\cite{press2023selfask,jin2025searchr1} issues queries and reads the returned pages. It must then make one decision over and over: is the evidence it holds \emph{enough} to answer, should it search again, or should it decline?
This decision sits at the core of retrieval-augmented AI on the Web. It governs when the assistant stops querying, what it shows its reader, and when it tells the user that the evidence found does not support an answer.
Answering too early produces confident answers that the retrieved pages do not support~\cite{joren2025sufficient}. Searching too long wastes queries and latency~\cite{xie2026oversearching,zhang2026tosearch}, and declining without saying what is missing gives the user nothing to act on.
Retrieval-augmented generation (RAG)~\cite{lewis2020rag} makes this decision implicit: whether the context sufficed is left to the reader.

Recent work makes the judgement explicit but keeps it inside a model.
A sufficient-context autorater asks a strong model whether a context could answer the question and combines its label with self-rated confidence to abstain~\cite{joren2025sufficient}. Learned or prompted stopping rules decide when an agent has searched enough~\cite{roh2026halt,li2026s2g}.
These judgements are made by a model about a context as a whole or about a flat list of claims.
Multi-hop questions, however, are conjunctions: \emph{which film was released first} needs both release dates, and \emph{when did X's mother die} needs the mother's identity before her date of death.
A context that contains every entity of the question and most of the facts can still miss one conjunct, and nothing guarantees that the judgement notices which one.
Making sufficiency explicit poses three challenges. The requirements of a question are not given and must be derived from its text. A retrieved page must be linked to the requirement it states, not to a page about a related entity. The evaluation needs unanswerable questions whose deciding fact is truly absent from the corpus.

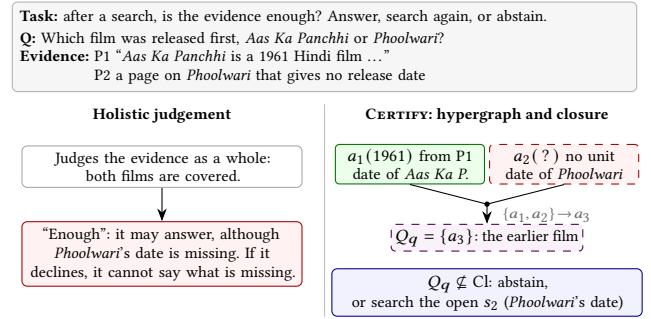
\begin{figure}[t]
\centering
\begin{tikzpicture}[
  font=\scriptsize, >=Stealth,
  band/.style={draw=gray!60, fill=gray!7, rounded corners=2pt, inner sep=3pt, align=left},
  head/.style={font=\scriptsize\bfseries, inner sep=1pt},
  box/.style={draw=gray!70, rounded corners=2pt, inner sep=2.2pt, align=center},
  atom/.style={draw=black!60, rounded corners=1.5pt, inner sep=1.6pt, align=center, fill=white},
  got/.style={atom, draw=green!45!black, fill=green!8},
  miss/.style={atom, draw=red!65!black, dashed, fill=red!5},
  tgt/.style={atom, draw=violet!60!black, dashed, fill=violet!5},
  bad/.style={box, draw=red!65!black, fill=red!5},
  good/.style={box, draw=blue!55!black, fill=blue!5}]
\node[band, text width=8.05cm, anchor=north] (task) at (4.15,0) {%
  \textbf{Task:} after a search, is the evidence enough? Answer, search again, or abstain.\\[1pt]
  \textbf{Q:} Which film was released first, \emph{Aas Ka Panchhi} or \emph{Phoolwari}?\\
  \textbf{Evidence:} P1 ``\emph{Aas Ka Panchhi} is a 1961 Hindi film \dots''\\
  \hphantom{\textbf{Evidence:}} P2 a page on \emph{Phoolwari} that gives no release date};
\coordinate (top) at ($(task.south)+(0,-1.6mm)$);
\node[head, anchor=north] (hl) at (top -| 2.0,0) {Holistic judgement};
\node[box, text width=3.55cm, anchor=north] (lj) at ($(hl.south)+(0,-2.6mm)$) {Judges the evidence as a whole:\\ both films are covered.};
\node[bad, text width=3.55cm, anchor=north] (lv) at ($(lj.south)+(0,-4.4mm)$) {``Enough'': it may answer, although\\ \emph{Phoolwari}'s date is missing. If it\\ declines, it cannot say what is missing.};
\draw[->] (lj) -- (lv);
\node[head, anchor=north] (hr) at (top -| 6.3,0) {\method{}: hypergraph and closure};
\node[got, text width=1.86cm, anchor=north] (a1) at ($(hr.south)+(-1.02,-2.4mm)$) {$a_1(1961)$ from P1\\ date of \emph{Aas Ka P.}};
\node[miss, text width=1.86cm, anchor=north] (a2) at ($(hr.south)+(1.02,-2.4mm)$) {$a_2(\,?\,)$ no unit\\ date of \emph{Phoolwari}};
\coordinate (j) at ($(a1.south)!0.5!(a2.south)+(0,-2.6mm)$);
\draw (a1.south) -- (j); \draw (a2.south) -- (j); \fill (j) circle (1pt);
\node[font=\tiny, anchor=west, text=black!60] at ($(j)+(0.9mm,-1.4mm)$) {$\{a_1,a_2\}\!\to\! a_3$};
\node[tgt, anchor=north] (q) at ($(j)+(0,-2.6mm)$) {$Q_q=\{a_3\}$: the earlier film};
\draw[->] (j) -- (q);
\node[good, text width=3.95cm, anchor=north] (rv) at ($(q.south)+(0,-2.2mm)$) {\mbox{$Q_q\not\subseteq\mathrm{Cl}$: abstain,}\\ or search the open $s_2$ (\emph{Phoolwari}'s date)};
\draw[gray!50] (top -| 4.15,0) -- (4.15,0 |- rv.south);
\end{tikzpicture}
\caption{The decision we study, on evidence that lacks one of two dates, a gap that a holistic judgement can miss. \method{} joins both dates in one hyperedge, so it answers only when both are stated and otherwise names the missing one.}
\label{fig:overview}
\end{figure}

We treat sufficiency as a structural property of the evidence.
For a query we build a \emph{requirement hypergraph}: each directed hyperedge says that a set of facts jointly yields another fact. A comparison is then one hyperedge with two required tail facts, and a bridge question is a chain of hyperedges whose later requirements are instantiated with earlier answers.
A context is sufficient when the monotone Horn closure of the facts it states contains the answer.
This yields \method{}, a certificate-driven search procedure (Figure~\ref{fig:overview}) that searches one requirement at a time, instantiated with the values already found, and links each retrieved unit to the requirement it states. It returns an inclusion-minimal sufficient context with its derivation, or abstains and reports the open requirements.
The closure does three things without an additional model call: it decides whether to answer, returns the minimal context the reader answers from, and names the open requirement to search next.
On a bridge chain it checks that every lookup is supported, and the same closure decides joins, computed requirements and alternative derivations.

The certificate is deductive and certifies that the context states every required fact.
When the gold support is retrieved and the requirements and links are correct, \method{} is sound, and complete on the retrieval event up to the entity check, and its context is inclusion-minimal.
A rule reading only the tail-to-head projection of the requirements cannot tell ``$a$ and $b$ yield $c$'' from ``$a$ or $b$ yields $c$''.

Pooled over 2WikiMultiHopQA, MuSiQue and HotpotQA with search-checked twins, \method{} decides significantly more accurately than search agents, graph retrieval and sufficiency judges in default configurations. On the same evidence, it matches a fast proprietary model's sufficiency label without an additional model call. Answers from its certificate are \nAsPooledCert\% correct against \nAsPooledPool\% from the whole pool. Over 21 million Wikipedia passages~\cite{karpukhin2020dpr}, entity anchoring raises its decision accuracy by 4 to 13 points on gold-verified pairs.

\paragraph{Contributions}
\begin{itemize}\setlength{\emergencystretch}{1.5em}
\item \emph{Sufficiency as Horn closure.} We decide sufficiency by the Horn closure of a per-query requirement hypergraph, with no model call once the evidence is linked. With correct retrieval, requirements and links, it is sound, complete up to the entity check and inclusion-minimal. Rules reading only pairwise projections provably cannot tell a conjunction from a disjunction.
\item \emph{\method{}, requirement-directed search.} Each search targets one requirement, instantiated with the values found, and anchored links supply its value. When the evidence does not suffice, \method{} abstains and names the open requirement, the deleted one in \nLocHit\% of traceable twins.
\item \emph{Evaluation with search-checked twins.} Questions are paired with search-checked unanswerable twins, and controls separate the evidence from the decision rule. With exact requirements \method{} recovers every enabling set in agent logs and every provable ProofWriter answer.
\end{itemize}

\section{Related Work}
\label{sec:related}

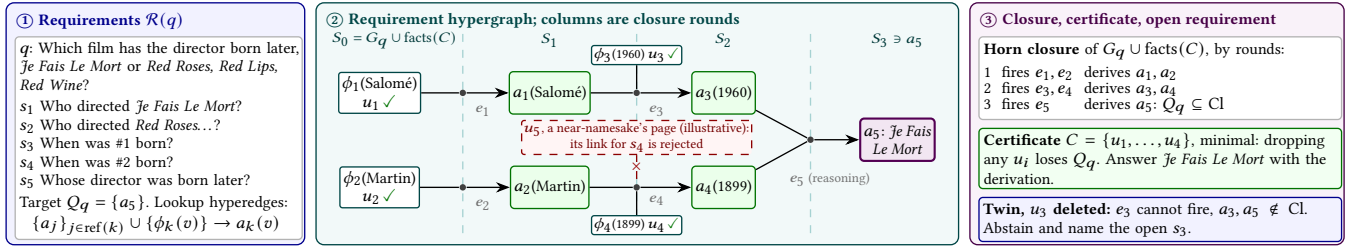
\begin{figure*}[t]
\centering
\begin{tikzpicture}[
  font=\scriptsize, >=Stealth, line width=0.5pt,
  panel/.style={draw=#1!55!black, fill=#1!5, rounded corners=3pt, inner sep=0pt},
  head/.style={font=\scriptsize\bfseries, text=#1!55!black, anchor=north west, inner sep=0pt},
  atom/.style={draw=black!60, fill=white, rounded corners=1.5pt, inner sep=1.3pt, align=center, minimum height=5.6mm},
  given/.style={atom, fill=gray!12},
  got/.style={atom, draw=green!45!black, fill=green!8},
  tgt/.style={atom, draw=violet!65!black, fill=violet!10, line width=0.8pt},
  fact/.style={draw=teal!55!black, fill=white, rounded corners=1pt, inner sep=0.9pt, font=\tiny, align=center},
  rej/.style={draw=red!55!black, fill=red!4, dashed, rounded corners=1pt, inner sep=0.9pt, font=\tiny, align=center, text=red!45!black},
  jn/.style={circle, fill=black!75, inner sep=0.85pt},
  el/.style={font=\tiny, text=black!55},
  card/.style={draw=black!30, fill=white, rounded corners=1.5pt, inner sep=1.8pt, align=left},
  okc/.style={card, draw=green!45!black, fill=green!5},
  noc/.style={card, draw=blue!55!black, fill=blue!5}]
\node[panel=blue, minimum width=3.95cm, minimum height=3.2cm, anchor=north west] (PA) at (0,0) {};
\node[panel=teal, minimum width=8.5cm, minimum height=3.2cm, anchor=north west] (PB) at (4.1,0) {};
\node[panel=violet, minimum width=4.95cm, minimum height=3.2cm, anchor=north west] (PC) at (12.75,0) {};
\node[head=blue] at (0.12,-0.13) {\textcircled{\tiny1} Requirements $\mathcal R(q)$};
\node[head=teal] at (4.22,-0.13) {\textcircled{\tiny2} Requirement hypergraph; columns are closure rounds};
\node[head=violet] at (12.87,-0.13) {\textcircled{\tiny3} Closure, certificate, open requirement};
\node[card, text width=3.67cm, anchor=north west] at (0.12,-0.45) {%
  $q$: Which film has the director born later, \emph{Je Fais Le Mort} or \emph{Red Roses, Red Lips, Red Wine}?\\[1pt]
  $s_1$ Who directed \emph{Je Fais Le Mort}?\\
  $s_2$ Who directed \emph{Red Roses\dots}?\\
  $s_3$ When was \#1 born?\\
  $s_4$ When was \#2 born?\\
  $s_5$ Whose director was born later?\\[1pt]
  Target $Q_q=\{a_5\}$. Lookup hyperedges:\\ \hspace*{1.5mm}$\{a_j\}_{j\in\mathrm{ref}(k)}\cup\{\phi_k(v)\}\to a_k(v)$};
\begin{scope}[xshift=2mm]
\foreach \x/\t in {4.95/{$S_0=G_q\cup\mathrm{facts}(C)$}, 7.0/{$S_1$}, 9.3/{$S_2$}, 11.6/{$S_3\ni a_5$}}
  \node[font=\tiny, text=teal!45!black] at (\x,-0.5) {\t};
\foreach \x in {5.85, 8.15, 10.45} \draw[teal!35, dashed] (\x,-0.42) -- (\x,-3.12);
\node[fact, font=\scriptsize, inner sep=1.3pt] (gA) at (4.75,-1.2) {$\phi_1$(Salomé)\\ $u_1$\,\textcolor{green!45!black}{\checkmark}};
\node[fact, font=\scriptsize, inner sep=1.3pt] (gB) at (4.75,-2.45) {$\phi_2$(Martin)\\ $u_2$\,\textcolor{green!45!black}{\checkmark}};
\node[got] (a1) at (7.0,-1.2) {$a_1$(Salomé)};
\node[got] (a2) at (7.0,-2.45) {$a_2$(Martin)};
\node[got] (a3) at (9.3,-1.2) {$a_3$(1960)};
\node[got] (a4) at (9.3,-2.45) {$a_4$(1899)};
\node[tgt] (a5) at (11.6,-1.82) {$a_5$: \emph{Je Fais}\\ \emph{Le Mort}};
\foreach \t/\h/\jx/\jy/\e in {gA/a1/5.85/-1.2/e_1, gB/a2/5.85/-2.45/e_2} {
  \node[jn] (j\h) at (\jx,\jy) {};
  \draw (\t) -- (j\h); \draw[->] (j\h) -- (\h);
  \node[el, anchor=north west] at ($(j\h)+(0.05,-0.05)$) {$\e$};
}
\foreach \t/\h/\jx/\jy/\f/\u/\side/\e in {a1/a3/8.15/-1.2/{$\phi_3$(1960)}/u_3/1/e_3, a2/a4/8.15/-2.45/{$\phi_4$(1899)}/u_4/-1/e_4} {
  \node[jn] (j\h) at (\jx,\jy) {};
  \node[fact] (f\h) at (\jx,\jy+\side*0.5) {\f\ $\u$\,\textcolor{green!45!black}{\checkmark}};
  \draw (\t) -- (j\h); \draw (f\h) -- (j\h); \draw[->] (j\h) -- (\h);
  \node[el, anchor=north west] at ($(j\h)+(0.05,\side*-0.03-0.02)$) {$\e$};
}
\node[rej] (rj) at (8.15,-1.79) {$u_5$, a near-namesake's page (illustrative):\\ its link for $s_4$ is rejected};
\draw[red!55!black, dashed] (rj.south) -- ($(ja4)+(0,0.09)$) node[pos=0.45, red!60!black, font=\tiny, fill=teal!5, inner sep=0.2pt] {$\times$};
\node[jn] (j5) at (10.45,-1.82) {};
\draw (a3) -- (j5); \draw (a4) -- (j5); \draw[->] (j5) -- (a5);
\node[el, anchor=north west, fill=teal!5, inner sep=0.4pt] at (10.15,-2.23) {$e_5$ (reasoning)};
\end{scope}
\node[card, text width=4.67cm, anchor=north west] (tr) at (12.87,-0.45) {%
  \textbf{Horn closure} of $G_q\cup\mathrm{facts}(C)$, by rounds:\\[0.5pt]
  \begin{tabular}{@{}l@{\ \ }l@{\ \ }l@{}}
  1 & fires $e_1,e_2$ & derives $a_1,a_2$\\[-1pt]
  2 & fires $e_3,e_4$ & derives $a_3,a_4$\\[-1pt]
  3 & fires $e_5$ & derives $a_5$: $Q_q\subseteq\mathrm{Cl}$
  \end{tabular}};
\node[okc, text width=4.67cm, anchor=north west] (ok) at ($(tr.south west)+(0,-1.0mm)$) {%
  \textbf{Certificate} $C=\{u_1,\dots,u_4\}$, minimal: dropping any $u_i$ loses $Q_q$. Answer \emph{Je Fais Le Mort} with the derivation.};
\node[noc, text width=4.67cm, anchor=north west] (no) at ($(ok.south west)+(0,-1.0mm)$) {%
  \textbf{Twin, $u_3$ deleted:} $e_3$ cannot fire, $a_3,a_5\notin\mathrm{Cl}$. Abstain and name the open $s_3$.};
\end{tikzpicture}
\caption{\method{} on a complex case: requirement hypergraph, Horn closure by rounds, minimal certificate and the twin's open requirement.}
\label{fig:framework}
\end{figure*}

\paragraph{Deciding sufficiency and abstaining in RAG}
Sufficient Context labels whether a context could answer a question with a prompted autorater and combines the label with self-rated confidence to abstain~\cite{joren2025sufficient}. Identify-then-verify prompting lists what a context lacks~\cite{jain2026missing}.
S2G-RAG asks a judge for sufficiency and for the missing items that drive the next query~\cite{li2026s2g}. MuSiQue's step-execution model judges a context sufficient when it suffices for every hop~\cite{trivedi2022musique}. HALT stops a search agent once a verifier matches every expected hop claim~\cite{roh2026halt}.
Uncertainty-driven retrieval~\cite{jiang2023flare}, selective answering~\cite{kamath2020selective}, training to refuse unsupported answers~\cite{song2025trustalign} and retrieval evaluators that trigger corrective search~\cite{yan2024crag} all judge a context with a model. \method{} derives sufficiency from explicit requirement hyperedges and returns a minimal context with its derivation.
Benchmarks probe abstention on unanswerable or misleading requests~\cite{rajpurkar2018squad,patidar2023grailqability,chen2024rgb,tang2024multihoprag,glockner2025neoqa,song2026webdetective,yang2024cragbench,kirichenko2025abstentionbench,muhamed2026refusalbench,peng2025uaeval4rag,liu2026crumqs}, and same-name entities mislead retrieval~\cite{chen2021amber,lee2024ambigdocs}. Closest to our twins are the deletion-based contrasts of DiRe and MuSiQue~\cite{trivedi2020dire,trivedi2022musique} and removing supporting evidence from the context~\cite{atanasova2022insufficient,lee2021pseudoevidentiality,gui2026faithrl,ming2025faitheval}. Our twins delete the evidence from the corpus and check by search that it is gone.
Conformal guarantees bound risk over a query distribution~\cite{kang2024crag,li2024traq}, certified robustness the effect of corrupted passages~\cite{xiang2026robustrag}. Our certificate is deductive and per query.

\paragraph{Iterative retrieval and search agents}
Interleaved retrieval and reasoning~\cite{trivedi2023ircot}, complexity routing~\cite{jeong2024adaptive}, reasoning models that search~\cite{li2025searcho1} and agents trained to search~\cite{jin2025searchr1} let the model decide when it has searched enough. So do models trained to abstain or to stop searching~\cite{zhang2026awarl,liu2026bapo,zhai2026abstainr1,zhang2026tosearch,sun2026autosearch}, self-reflection~\cite{asai2024selfrag}, step verification~\cite{zhao2023verifyedit,li2024chainofknowledge} and learned stopping~\cite{qi2021irrr,yang2025simrag}. More search can hurt abstention on unanswerable queries~\cite{xie2026oversearching}.
Decomposition into sub-questions retrieved or grounded one by one~\cite{qi2019golden,wolfson2020break,khot2023decomposed,xu2024searchchain,zhuang2024efficientrag,shi2024genground,zhu2025chainrag} or into question trees~\cite{cao2023probtree,zhang2023roht} and planners that order sub-questions in a dependency graph or re-plan after evaluating each step~\cite{chen2026logicrag,xiao2025lag,zhu2026reap,chen2025mindsearch} decide what to search next. Search-in-the-Chain verifies each step against retrieved passages~\cite{xu2024searchchain}, and SR-RAG checks steps and chain completeness with a symbolic verifier~\cite{wang2026srrag}. QACheck, Think-on-Graph, ProgramFC and FOLK let an LLM judge whether sub-answers or reasoning paths suffice, or combine sub-results logically~\cite{pan2023qacheck,sun2024thinkongraph,pan2023programfc,wang2023folk}. In \method{}, a closure over per-requirement links decides without a model call, returning a minimal context or the open requirement to search next.

\paragraph{Graphs and hypergraphs for retrieval and memory}
Graph-based RAG links passages through entities and propagates relevance~\cite{gutierrez2025hipporag2}. Query-specific hyper-relational graphs~\cite{panda2024holmes} and hypergraph RAG keep $n$-ary facts as hyperedges~\cite{luo2025hypergraphrag,lien2026hyperrag,zai2026proh,wang2026hgrag,feng2026hyperrag}, hypergraph memories group long interactions~\cite{zhou2026hgmem}, and recent work learns hyperedge retrievers~\cite{zhao2026hyperreco}.
Like them, \method{} keeps $n$-ary facts as hyperedges, but its directed requirement hyperedges decide whether the evidence found suffices, which a rule reading only their pairwise projection cannot do reliably.

\paragraph{Directed hypergraphs and Horn reasoning}
Directed hypergraphs and B-connectivity~\cite{gallo1993directed} and linear-time Horn satisfiability~\cite{dowling1984linear} are classical. Why-provenance attaches derivations to query answers~\cite{buneman2001why,calautti2024why}, query completeness decides whether an incomplete database still returns every answer~\cite{levy1996complete,razniewski2011completeness}, and rationales and context attributions are sought for predictions and RAG answers~\cite{deyoung2020eraser,cohenwang2024contextcite}. Our pruning is the classical deletion filter for minimal sets under a monotone predicate~\cite{chinneck1991locating,marquessilva2013msmp}, as in sufficient input subsets~\cite{carter2019sufficient} and minimal evidence groups~\cite{li2025minimal}.
Proofs over stated facts~\cite{tafjord2021proofwriter,dalvi2021entailmentbank}, including deduction that names a missing premise~\cite{sprague2022incomplete}, backward chaining or unification over facts that an LLM supplies or retrieves~\cite{kazemi2023lambada,tafjord2022entailer,weir2024nellie,weber2019nlprolog,lee2025symba} and solvers over LLM-extracted facts~\cite{pan2023logiclm} derive answers. So does concurrent work with deterministic operators over a knowledge lattice and abstention~\cite{ren2026latweave} or with Prolog written from passages~\cite{gann2026nesyrag}. We use the closure for a different decision, whether a retrieved context suffices.

\section{Evidence, Requirements and Sufficiency}
\label{sec:formal}

A retrieval-augmented system answers a query $q$ from a collection $\Units$ of \emph{evidence units}: passages of Web documents for a search assistant, or turns of an interaction log for an agent with memory.
A \emph{context} is a finite set $\Ctx\subseteq\Units$ that the system passes to its reader.
We make precise when a context suffices to answer $q$.

\paragraph{Requirement hypergraphs}
Let $\Atoms_q$ be a finite set of \emph{atoms}, the propositions that matter for $q$: the value of an intermediate entity, a date, the state of an object at some time.

\begin{definition}[Requirement hypergraph]\label{def:req}
A requirement hypergraph is $\Rq=(\Atoms_q,\Rules_q,\Target_q,G_q)$, where every $r\in\Rules_q$ is a directed hyperedge $r=(\Req(r),h(r),\Exc(r))$ with required tail atoms $\Req(r)\subseteq\Atoms_q$, a head atom $h(r)\in\Atoms_q$ and excluded atoms $\Exc(r)\subseteq\Atoms_q$ (empty in both instantiations below).
Here $\Target_q\subseteq\Atoms_q$ is the set of target atoms and $G_q\subseteq\Atoms_q$ the atoms given with the query.
The tail size $\delta^T_r=|\Req(r)|$ counts the facts it conjoins.
\end{definition}

\begin{definition}[Evidence]\label{def:evidence}
Every unit $u\in\Units$ \emph{states} a set of atoms $\facts(u)\subseteq\Atoms_q$ and \emph{witnesses} a set of hyperedges $\rules(u)\subseteq\Rules_q$.
The hyperedges $\Rules_q^{0}\subseteq\Rules_q$ witnessed by no unit belong to the query.
Both maps extend additively to contexts: $\facts(\Ctx)=\bigcup_{u\in\Ctx}\facts(u)$ and $\rules(\Ctx)=\Rules_q^{0}\cup\bigcup_{u\in\Ctx}\rules(u)$.
\end{definition}

\begin{definition}[Horn closure]\label{def:closure}
For atoms $S$ and hyperedges $E$, let $\Phi_{S,E}(X)=S\cup\{h(r): r\in E,\ \Req(r)\subseteq X\}$.
The closure $\Cl(S,E)$ is the least fixpoint of $\Phi_{S,E}$.
Excluded atoms do not take part in firing, so that a least fixpoint exists.
\end{definition}

\begin{definition}[Sufficiency and violation]\label{def:suff}
A context $\Ctx$ is \emph{sufficient} for $q$ if $\Target_q\subseteq\Cl\big(G_q\cup\facts(\Ctx),\ \rules(\Ctx)\big)$.
A hyperedge $r$ is \emph{violated} with respect to a set $U$ of units (in \method{}, the linked units of the retrieval pool) if $\Exc(r)\cap\facts(U)\neq\emptyset$.
It is also violated if its head $h(r)$ has a rival value whose support in $U$ is at least that of the chosen value.
Violations change the decision to answer, never the closure.
Since a new unit can create a rival, the decision to answer is not monotone in the context, although sufficiency is.
\end{definition}
\begin{lemma}[Closure operator]\label{lem:closure}
$\Phi_{S,E}$ is monotone in $X$, $S$ and $E$.
Hence $\Cl(S,E)$ exists, $S\subseteq\Cl(S,E)$, $\Cl$ is monotone in both arguments and idempotent, and sufficiency is monotone in the context: $\Ctx\subseteq\Ctx'$ and $\Ctx$ sufficient imply $\Ctx'$ sufficient.
Forward chaining with one counter per hyperedge computes $\Cl(S,E)$ in $O\big(|S|+\sum_{r\in E}(\delta^T_r+1)\big)$ time~\cite{dowling1984linear}.
\end{lemma}

\paragraph{Two instantiations}
\emph{Multi-hop questions over Web documents.}
A question is decomposed into requirements $s_1,\dots,s_m$, each a single-fact question that may refer to the answers of earlier requirements $\mathrm{ref}(k)\subseteq\{1,\dots,k-1\}$ (``When did \#1 die?'').
Atom $a_k(v)$ reads ``requirement $k$ has value $v$''.
A \emph{lookup} requirement needs a passage: a passage $u$ states $\phi_k(v)\in\facts(u)$ if it states that the answer to $s_k$, instantiated with the chosen values of its references (one instantiation), is $v$, and the query contributes the hyperedge $\{a_j(v_j)\}_{j\in\mathrm{ref}(k)}\cup\{\phi_k(v)\}\to a_k(v)$.
A \emph{reasoning} requirement, such as choosing the earlier of two dates, is computed from earlier derived facts: $\{a_j(v_j)\}_{j\in\mathrm{dep}(k)}\to a_k(v)$ with $\mathrm{dep}(k)\supseteq\mathrm{ref}(k)$.
The target is $a_m(v)$ for the last requirement $m$ and its best-supported value $v$.
Sufficiency and minimality are relative to that instantiation and value, and rival values are left to the violation rule.
Comparison questions therefore carry a genuine conjunction, $\delta^T_r=2$, that no single passage satisfies.
\emph{Interaction logs of agents.}
For an environment with states $x_0,\dots,x_T$, an atom is a pair $(f,k)$: fact $f$ during its $k$-th maximal validity interval.
Turn $t$ executes an action $\alpha_t$ with preconditions $\mathrm{Pre}(\alpha_t)$ and witnesses the hyperedges $\{(g,\kappa(g,t-1)):g\in\mathrm{Pre}(\alpha_t)\}\to(w,\kappa(w,t))$ for every fact $w$ true at $t$ but not at $t-1$, where $\kappa(f,t)$ indexes the interval of $f$ that contains $t$.
Each interval is produced only by the turn that opens it.
Deleted facts only close intervals and generate no hyperedge, so nothing is ever underived and the closure stays monotone.

\paragraph{What projections lose}
Many graph-based retrieval and memory systems keep a pairwise view of this structure.
The \emph{tail-to-head projection} of a hyperedge set is $\Proj(\Rules)=\{(v,h(r)): r\in\Rules,\ v\in\Req(r)\}$.
The \emph{tail shuffle} $\tau$ replaces every $r$ with $\delta^T_r\ge2$ by the singleton hyperedges $\{v\}\to h(r)$, $v\in\Req(r)$.
It keeps $\Proj$ unchanged and destroys only the conjunction.
For retrieval over documents, let $\mathbf B\in\{0,1\}^{|\mathcal A|\times|\Units|}$ be the incidence between entities $\mathcal A$ and units, and let its 2-section be $\Sec(\mathbf B)=\{\{a,b\}:\exists u,\ \mathbf B_{au}\mathbf B_{bu}=1\}$.

An \emph{instance} $\mathcal M$ is a query with its requirements, evidence maps and context, $\suff(\mathcal M)$ is its sufficiency label, $\tau\mathcal M$ applies $\tau$ to its hyperedges, and $\tau_\#\mathcal D$ is the law of $\tau\mathcal M$ for $\mathcal M\sim\mathcal D$.
The following proposition is why \method{} decides on directed hyperedges.

\begin{proposition}[Directed blindness]\label{prop:directed}
Let $D$ be any decision rule whose input is the query, the contents of the atoms, the facts the context states and $\Proj(\Rules_q)$.
Then $D$ returns the same label on an instance $\mathcal M$ and on its tail-shuffled twin $\tau\mathcal M$.
For any distribution $\mathcal D$ over instances, on the balanced mixture of $\mathcal D$ and $\tau_\#\mathcal D$,
$\Pr[D\text{ errs}]\ge\tfrac12\Pr_{\mathcal M\sim\mathcal D}[\suff(\mathcal M)\neq\suff(\tau\mathcal M)]$.
\end{proposition}

\begin{proposition}[Unit blindness of the 2-section]\label{prop:undirected}
Let $u^\star\in\Units$ be a unit every entity pair of which co-occurs in another unit.
Then $\Sec(\mathbf B)$ is the same for $\Units$ and $\Units\setminus\{u^\star\}$, so every retrieval or decision rule whose input is the query, entity contents and $\Sec(\mathbf B)$ gives the same output on both collections.
If $\Units$ is sufficient and $u^\star$ is the only unit stating an atom outside $G_q$ that no hyperedge derives and every derivation of $\Target_q$ uses, sufficiency flips from true to false.
On the balanced mixture of such pairs any such rule errs with probability at least $\tfrac12$.
\end{proposition}

This proposition is why \method{} links individual units to requirements. Both propositions are scoped two-point separation arguments (Appendix~\ref{app:proofs}).
A star expansion with one node per hyperedge, typed incidences and tail sizes $\delta^T_r$ is information-equivalent to the hypergraph, and passage-node graphs keep unit identity.
For them the question is whether the decision uses the information, as \method{} does by construction (the directed witness is classical~\cite{gallo1993directed}).

\section{Certified Sufficiency}
\label{sec:method}

\method{} turns sufficiency into search (Algorithm~\ref{alg:certify}, Figure~\ref{fig:framework}).

\begin{algorithm}[t]
\caption{\method{}: certificate-driven search}\label{alg:certify}
\KwIn{query $q$; unit collection $\Units$; retriever; linker; search schedule $(\ell_1,\dots,\ell_L)$}
\KwOut{answer with certificate $(\Ctx,\text{derivation})$, or abstention with open requirements}
build $\Rq$ \tcp*{oracle annotation, engine rules, or LLM decomposition of $q$ alone}
$\mathit{pool}\gets\emptyset$; $\mathit{open}\gets\emptyset$\;
\ForEach{requirement $s_k$ in topological order}{
  \lIf{a reference of $s_k$ is unresolved}{$\mathit{open}\gets\mathit{open}\cup\{k\}$; \textbf{continue}}
  instantiate $s_k$ with the resolved values of its references\;
  \lIf{$s_k$ is a reasoning requirement}{derive $a_k(v)$ from $\mathrm{dep}(k)$ and the reader's value $v$ (no value: add $k$ to $\mathit{open}$); \textbf{continue}}
  \lFor{$i=1,\dots,L$ until some unit states $\phi_k(\cdot)$}{retrieve $\ell_i$ new units for $s_k$, link each to $\phi_k(\cdot)$ and add it to $\mathit{pool}$}
  \lIf{no unit states $\phi_k(\cdot)$}{$\mathit{open}\gets\mathit{open}\cup\{k\}$}
}
$S\gets\Cl\big(G_q\cup\facts(\mathit{pool}),\rules(\mathit{pool})\big)$\;
\lIf{$\Target_q\not\subseteq S$}{\Return abstain with the requirements in $\mathit{open}$}
\lIf{a lookup step of the derivation has no accepted link}{\Return abstain with that requirement}
$\Ctx\gets$ first producers of the backward derivation of $\Target_q$ over accepted links, in retrieval order\;
\lForEach{$u\in\Ctx$ in retrieval order}{\textbf{if} $\Ctx\setminus\{u\}$ is sufficient for $\Target_q$ over accepted links \textbf{then} $\Ctx\gets\Ctx\setminus\{u\}$}
\lIf{violation rule and a hyperedge of the derivation is violated}{\Return abstain with the violated requirement}
\Return reader$(q,\Ctx,\text{derivation})$ with certificate $(\Ctx,\text{derivation})$
\end{algorithm}

\paragraph{Requirements, retrieval and linking}
On Web question answering, $\Rq$ comes from decomposing $q$ with a language model that sees only $q$ (\emph{deployed}), or from dataset annotations (\emph{oracle}, reported separately).
Retrieval is per requirement: the query sent to the search engine is the instantiated requirement, so each search targets a single missing fact, and the open requirements decide what to search next.
A reasoning requirement takes the value the reader computes in one call from every fact derived before it, its $\mathrm{dep}(k)$, and stays open without one.
Linking asks a reader whether one unit states the answer to one instantiated requirement.
Different values stated for the same requirement are grouped (normalised string match, containment or token F1 $\ge0.6$, greedily in retrieval order against each group's first string, which names it).
The \emph{support} of a value is the number of linked units that state it.
Atoms $a_k(v)$ exist only for values that some linked unit states, so $\Atoms_q$ is finite and determined by the pool.
The search for a requirement stops after the first round that links a unit, every unit of that round being judged.
The chosen value has the largest support among linked units (ties: the earliest), and a rival that ties it violates the requirement.

\paragraph{Entity anchoring and verification}
Most false links connect a requirement to a passage about another entity, often a near-namesake, that states a value for the wrong subject.
A linked unit is \emph{anchored} to a lookup requirement when its title names the requirement's subject: every word of the title (disambiguating parenthetical removed) occurs in the instantiated requirement. Title matching also selects documents in fact verification~\cite{nie2019combining}.
Consider a lookup step whose value has no anchored supporting link.
It is \emph{verified} when the reader, shown the first two supporting passages in turn, answers that one of them gives the value for exactly the entity the requirement names.
\method{} accepts the derivation only when the chosen value of every lookup step is anchored or verified (its link is \emph{accepted} by this \emph{entity check}).
It computes the context over accepted links alone and lets the best-supported value stand when a rival ties it.

\paragraph{The Horn-closure certificate}
After the search, forward chaining over the linked pool decides whether $\Target_q$ is derived, and the entity check is applied to each lookup step of the derivation.
Over accepted links, a unit $u$ states only the atoms $\facts_{\mathrm{acc}}(u)\subseteq\facts(u)$ whose link passes the check for their step. There the backward derivation takes the earliest retrieved unit stating each of its atoms, and the deletion filter drops every unit whose removal keeps $\Target_q$ derived.
The certificate is this minimal context $\Ctx$ with its derivation, which gives each requirement's value, passage and hyperedge (Figure~\ref{fig:framework}).
On abstention \method{} reports the open requirements. Those whose references are resolved are the facts no linked passage states, the next to search for, and the others depend on them (after a failed check, it reports the failing lookup).
On a bridge chain of lookups the closure is the step-wise rule ``every lookup has an accepted link''. In general it requires the target itself to be derived and checks only the lookups its derivation uses, so the same check decides joins, computed requirements and alternative derivations.
Given the links, the certificate attests that the context states every required fact (under the violation rule, also that no rival value is as well supported), not that the answer is correct.

\paragraph{Properties}
Write $\mathcal R^*(q)$ and $\facts^*,\rules^*$ for the true requirements and evidence maps, and $F_q\subseteq\Units$ for a sufficient \emph{gold support}.
The following proposition states when \method{}'s decision is exactly right.

\begin{proposition}[Conditional certificate]\label{thm:cert}
Assume three hypotheses. (H1) \emph{Gold retrieval}: the search for each requirement retrieves its gold units, so $F_q\subseteq\mathit{pool}$. (H2) \emph{Correct requirements}: $\Rq=\mathcal R^*(q)$, the true requirements instantiated with the values the pool states. (H3) \emph{Exact linking}: $\facts(u)=\facts^*(u)$ and $\rules(u)=\rules^*(u)$ for $u\in\mathit{pool}$, and each lookup requirement has one true value. Then:
\begin{enumerate}
\item \emph{Completeness on the retrieval event.} $\Target_q$ is derived, so \method{} abstains only on a violation or a lookup link that fails the entity check, never for lack of evidence.
\item \emph{Soundness.} If \method{} answers with context $\Ctx$, then $\Ctx$ is sufficient under $(\mathcal R^*,\facts^*,\rules^*)$ and violates no hyperedge of its derivation (under (H3) no rival value is linked, so this holds with rival values allowed).
\item \emph{Minimality.} The returned $\Ctx$ is sufficient, and no proper subset is sufficient over accepted links.
\item \emph{Uniqueness.} Suppose the pool derives $\Target_q$ by a single derivation (every atom derived as a head is stated by no unit and is the head of exactly one hyperedge that fires). If its stated atoms outside $G_q$ and witnessed hyperedges outside $\Rules_q^{0}$ each have exactly one producing unit in the pool, $\Ctx$ is the unique inclusion-minimal subset of the pool sufficient over accepted links and does not depend on the pruning order.
\item \emph{Cost.} Certification takes at most $|\Ctx_0|+1$ closure computations, where $\Ctx_0$ is the backward-derivation set, each linear in the size of the instance.
\end{enumerate}
\end{proposition}

The proof (Appendix~\ref{app:proofs}) uses only the closure's monotonicity, with minimality as for the deletion filter~\cite{chinneck1991locating,marquessilva2013msmp}.
Below we relax exact linking to accepted links and bound false certificates on twins.

\begin{corollary}[Accepted-link soundness]\label{cor:anchor}
Proposition~\ref{thm:cert}(2) holds for sufficiency with exact linking required only on accepted links. Assume correct requirements, and let a twin delete every unit that states a fact every derivation of $\Target_q$ uses. The certificate fires on it only if an accepted unit stating that fact survives the deletion or the derivation uses a false accepted link. Its false-certificate rate on twins is at most the rate of surviving accepted statements plus that of derivations with a false accepted link.
\end{corollary}
\begin{proof}
The context is computed over accepted links, so every lookup step of the returned derivation is supported by an accepted link of the context. The proof of (2) goes through with exact linking restricted to accepted links. On a twin, with correct requirements the derivation contains the deleted fact's requirement and derives it from some accepted unit. If every accepted link of the derivation is correct, the requirements are instantiated with true values and that unit states the fact, so it is a statement the deletion missed.
\end{proof}

\section{Experimental Setup}
\label{sec:setup}

\paragraph{Web question answering}
We use the pooled Wikipedia corpora of HippoRAG~2~\cite{gutierrez2025hipporag2} for 2WikiMultiHopQA~\cite{ho2020twowiki} and MuSiQue~\cite{trivedi2022musique}: 1{,}000 questions per dataset over 6{,}119 and 11{,}656 passages, in which every question competes with the evidence of all others.
The first 200 questions in their original order form the validation split, used for every choice of prompt, threshold and selective head. The remaining 800 form the test split, touched once per arm.
All arms share one dense retriever (Qwen3-Embedding-0.6B~\cite{qwen3embedding}) and one reader (Qwen3-8B~\cite{qwen3}, no thinking, greedy, served with vLLM~\cite{kwon2023vllm}). We use Qwen3-32B and DeepSeek V4.1 Flash where stated, and HotpotQA~\cite{yang2018hotpotqa} from the same release (twins delete the passage of one of two Gemini-2.5-Flash lookups).

\paragraph{Unanswerable twins}
\looseness=-1 For every question we build a twin whose evidence is deleted from the corpus. Unlike generated unanswerable questions~\cite{peng2025uaeval4rag,liu2026crumqs}, it differs from its original only in that evidence.
On MuSiQue we delete the supporting passages of the hops that MuSiQue's own unanswerable contrast removes~\cite{trivedi2022musique}, so the choice of deletion is not ours. On 2WikiMultiHopQA we delete the gold passage of one lookup hop chosen with a fixed seed.
On both we then search the remaining corpus for any other passage that states the deleted fact, with the gold bridge values and the 32B model as linker on the top 30 hits. We delete these as well (at most three), and questions whose fact is stated more often are dropped.
A question is kept only if neither Qwen3-8B nor Qwen3-32B answers it correctly without context in five samples (a memory filter).
This leaves 529 twins on 2WikiMultiHopQA (427 in test) and 826 on MuSiQue (661). The filter removes most guessable yes/no comparisons of 2WikiMultiHopQA, so conjunctions are tested mainly by bridge comparisons and MuSiQue requirements using two earlier answers.

\paragraph{Exact requirements}
ProofWriter's depth-5 test theories~\cite{tafjord2021proofwriter} and TextWorld logs~\cite{cote2018textworld} produced with AMA-Bench's protocol~\cite{zhao2026amabench} give exact requirement hypergraphs. In 959 episodes (671 for test), a query asks which turns enable an action, and its twin deletes one.

\paragraph{Baselines and controls}
\looseness=-1 Each baseline runs in its default configuration, with its authors' prompts, code or checkpoint, under our retriever, reader and corpus. Only HALT and S2G-RAG deviate: they abstain when their rule never accepts the evidence, the reading more favourable to them (App.~\ref{app:impl} gives every method's settings and the fixed seeds).
\emph{Single pass:} \emph{RAG} on the top-5 passages, \emph{Verbalised sufficiency} (the reader checks each requirement before answering) and \emph{HippoRAG~2}~\cite{gutierrez2025hipporag2} (graph-based retrieval).
\emph{Sufficiency judges:} \emph{Sufficient Context}~\cite{joren2025sufficient}, an autorater label (Qwen3-8B, Qwen3-32B or Gemini-2.5-Flash) combined with $P(\mathrm{True})$~\cite{kadavath2022know} in a selective head. \emph{Abstain-R1}~\cite{zhai2026abstainr1} is a model trained with reinforcement learning to answer or abstain (its 3B checkpoint, closed-book by default).
\emph{Search agents:} \emph{Iterative search}~\cite{press2023selfask,trivedi2023ircot}, \emph{HALT}~\cite{roh2026halt} (a verifier-based stopping rule on that agent) and \emph{S2G-RAG}~\cite{li2026s2g} (sufficiency-and-gap judging after each search). S2G-RAG prompts our reader in place of its distilled judge, and HALT's verifier takes our reader's current answer as its gold. \emph{Search-R1}~\cite{jin2025searchr1}, trained with reinforcement learning (its 7B checkpoint, and its 32B one with Qwen3-32B), never abstains, so a trajectory without an answer counts as an abstention. \emph{LogicRAG}~\cite{chen2026logicrag} orders sub-questions in a dependency graph, and \emph{SR-RAG}~\cite{wang2026srrag} checks each sub-answer with a Prolog prover over extracted triples (reimplemented).
\emph{Controls with \method{}'s inputs:} an LLM judge on \method{}'s pool and requirement list, and an LLM judge on its accepted link table (decompose-and-verify with an LLM decision). The Sufficient Context label on multi-step evidence (\method{}'s pool, or the iterative agent's or Search-R1's evidence) alone decides whether that evidence's answer stands, and Abstain-R1 reads \method{}'s pool.
\method{} uses the same reader in every role and the deployed decomposition, and certifies over accepted links with rival values allowed (no violation rule). There were eight candidate rules per regime: two rival rules, each with no, subject or subject-or-value anchoring, and two that also accept verified links. The 200 validation items per corpus selected subject anchoring or verification on the pooled corpora (decision accuracy \nSelPooled\%) and subject anchoring alone at Web scale (\nSelWeb\%).

\paragraph{Metrics}
\looseness=-1 Correctness is judged by Gemini-2.5-Flash against the reference answers only. DeepSeek V4.1 Flash as a second judge agrees on \nDjTWCert, \nDjMUCert{} and \nDjHOCert\% of \method{}'s answers and leaves the decision comparisons it covers unchanged. Gemini-2.5-Flash also audits links and twins. DeepSeek re-grades the links and audits the HotpotQA twins, and a blind annotator labels a sample.
On answerable questions we report accuracy (Acc, over all originals / over those no reader answers from memory) and accuracy among answered questions (SAcc). Costs are LLM calls (Calls, verification included) and thousands of input tokens (Tok) per question. On originals and twins we report coverage (Cov, share answered), the twins answered (Twin) and those answered with the original's gold answer (Gold), and abstention precision (AbsP). Decision accuracy (Dec) counts a correct answer on an original or an abstention on a twin, and always abstaining scores 50. $\Delta$Dec is \method{} minus a row, $\pm$ its paired-bootstrap standard error. In the result tables the best value is bold, the second best underlined and \method{}'s row shaded. Search-R1 rows use its own checkpoint (7B, or 32B with Qwen3-32B), and other rows use the setting's model in every role. 2W, MQ and HQ are the three corpora. SR-RAG and Abstain-R1 are compared with \method{}'s rebuilt 8B rows (App.~\ref{app:anchor}).
Decoding is greedy. Differences are tested with a paired bootstrap (10{,}000 resamples of pairs) and McNemar tests Holm-corrected over the arms of a setting. In the results, \emph{significantly}, \emph{above} and \emph{below} mean Holm-corrected $p<0.05$, and \emph{level} means not significant, even where the bootstrap alone is significant.

\section{Results}
\label{sec:results}

We ask five questions.
\textbf{RQ1}: how accurate are certified answers on answerable questions?
\textbf{RQ2}: does the certificate decide between answering and abstaining better than holistic sufficiency judges and stopping rules when evidence is missing?
\textbf{RQ3}: what does the certificate deliver beyond the decision, and at what cost?
\textbf{RQ4}: how often do the hypotheses behind the certificate's guarantee hold?
\textbf{RQ5}: with exact requirements, what do pairwise views of the evidence lose?

\begin{table}[t]
\caption{Accuracy and cost on answerable questions (test, 8B). \method{} has the highest SAcc on both corpora, and only HALT and SR-RAG make more LLM calls. $^\dagger$Offline indexing not counted.}
\label{tab:answerable}
\centering\footnotesize
\setlength{\tabcolsep}{1.05pt}
\begin{tabular}{lrrrrrrrr}
\toprule
 & \multicolumn{4}{c}{2WikiMultiHopQA} & \multicolumn{4}{c}{MuSiQue} \\
 & Acc$\uparrow$ & SAcc$\uparrow$ & Calls$\downarrow$ & Tok$\downarrow$ & Acc$\uparrow$ & SAcc$\uparrow$ & Calls$\downarrow$ & Tok$\downarrow$ \\ \midrule
RAG (top-5) & 31.8\,/\,22.5 & 50.2 & \textbf{1} & \textbf{0.9} & 17.9\,/\,16.3 & 41.8 & \textbf{1} & \textbf{0.8} \\
Verbalised sufficiency & 32.0\,/\,21.8 & 53.0 & \underline{2} & \underline{1.6} & 17.1\,/\,14.7 & 41.9 & \underline{2} & \underline{1.6} \\
Iterative search & 51.9\,/\,47.3 & 55.6 & 3 & 2.2 & 34.6\,/\,31.9 & 41.8 & 3 & 2.1 \\
S2G-RAG & 45.0\,/\,40.0 & 71.6 & 5 & 5.6 & 17.0\,/\,15.6 & 59.6 & 7 & 8.2 \\
Search-R1 & \underline{70.2}\,/\,\textbf{66.7} & \underline{72.1} & 4 & 3.4 & \textbf{47.2}\,/\,\textbf{45.7} & 51.9 & 4 & 3.5 \\
LogicRAG & 61.3\,/\,56.2 & 64.1 & 5 & 2.3 & 35.6\,/\,34.9 & 35.9 & 7 & 2.8 \\
SR-RAG & 65.2\,/\,57.4 & 66.2 & 23 & 19.9 & \underline{39.2}\,/\,\underline{36.9} & 40.9 & 23 & 18.3 \\
HippoRAG 2$^\dagger$ & 39.9\,/\,31.6 & 57.6 & \underline{2} & 4.1 & 19.2\,/\,17.9 & 41.5 & \underline{2} & 3.8 \\
HALT & 11.0\,/\,5.9 & 58.7 & 49 & 39.2 & 2.4\,/\,2.7 & \underline{63.3} & 60 & 49.3 \\
\ours{}\method{} & \textbf{73.0}\,/\,\underline{65.6} & \textbf{85.8} & 17 & 8.9 & 30.8\,/\,28.4 & \textbf{74.5} & 22 & 11.1 \\
\bottomrule
\end{tabular}
\end{table}

\subsection{RQ1: Answerable questions}
Answers produced through the requirement hypergraph over accepted links are the most reliable (Table~\ref{tab:answerable}, \nTWAnvSAcc\% and \nMUAnvSAcc\% correct), and \method{} answers the most 2WikiMultiHopQA originals correctly (\nTWAnvAcc\% of all).
MuSiQue often states a fact on a page titled by neither the subject nor the value, and there \method{} trades coverage for precision. It answers \nMUAnvAcc\% of all originals correctly (\nMULenientAcc\% without the entity check). Search-R1, SR-RAG, LogicRAG and the iterative agent answer more (\nMUIterAcc{} to \nMUSrAcc\%) but also answer most twins.
Decomposing from the question alone is level with the annotated decomposition (decision accuracy \nTWCertDec{} and \nMUCertDec\% against \nTWOracleDec{} and \nMUOracleDec\%, both under the violation rule and without the entity check, unlike Table~\ref{tab:selective}'s row).

\emph{Answer:} certified answers are the most reliable of those given.

\begin{table*}[t]
\caption{Answering or abstaining on paired originals and twins (test, 8B, $\pm$ standard error). \method{} is most accurate on 2WikiMultiHopQA and over two standard errors above every first-block baseline.}
\label{tab:selective}
\centering\footnotesize
\setlength{\tabcolsep}{3.95pt}
\begin{tabular}{lrrrrrrrrrrrr}
\toprule
 & \multicolumn{6}{c}{2WikiMultiHopQA} & \multicolumn{6}{c}{MuSiQue} \\
 & Cov & Twin$\downarrow$ & Gold$\downarrow$ & AbsP$\uparrow$ & Dec$\uparrow$ & $\Delta$Dec & Cov & Twin$\downarrow$ & Gold$\downarrow$ & AbsP$\uparrow$ & Dec$\uparrow$ & $\Delta$Dec \\
\midrule\multicolumn{13}{l}{\emph{Baselines, each on its own evidence}} \\
RAG (top-5) & 48.2 & 32.8 & 3.7 & 64.9 & 44.8 & 33.8$\pm$1.8 & 30.4 & 20.7 & 4.2 & 57.0 & 47.8 & 12.6$\pm$1.3 \\
Verbalised sufficiency & 44.7 & 29.0 & 3.5 & 64.2 & 46.4 & 32.3$\pm$1.8 & 26.6 & 16.2 & 2.4 & 57.1 & 49.2 & 11.2$\pm$1.2 \\
Iterative search & 88.5 & 83.1 & 13.6 & 73.5 & 32.1 & 46.6$\pm$1.8 & 74.2 & 67.2 & 10.1 & 63.6 & 32.4 & 28.1$\pm$1.4 \\
S2G-RAG & 42.3 & 22.5 & 4.0 & 67.1 & 58.8 & 19.9$\pm$1.7 & 18.0 & 9.4 & 2.4 & 55.3 & 53.1 & 7.3$\pm$1.1 \\
Search-R1 & 93.9 & 90.4 & 17.6 & 78.8 & 38.2 & 40.5$\pm$1.6 & 86.0 & 81.5 & 15.1 & \underline{65.9} & 32.1 & 28.4$\pm$1.2 \\
LogicRAG & 87.0 & 81.5 & 13.8 & 71.2 & 37.4 & 41.3$\pm$1.7 & 98.6 & 98.2 & 12.4 & \textbf{66.7} & 18.4 & 42.1$\pm$1.2 \\
SR-RAG & 97.5 & 97.2 & 13.1 & 57.1 & 30.1 & 49.1$\pm$1.6 & 96.6 & 97.1 & 8.8 & 42.2 & 19.9 & 40.7$\pm$1.2 \\
HippoRAG 2 & 51.4 & 33.0 & 4.0 & 68.9 & 49.3 & 29.4$\pm$1.8 & 33.4 & 23.3 & 3.8 & 57.5 & 47.3 & 13.2$\pm$1.3 \\
HALT & 6.6 & \textbf{1.9} & \textbf{0.2} & 52.5 & 52.0 & 26.7$\pm$1.4 & 2.4 & \textbf{0.9} & \textbf{0.3} & 50.8 & 50.9 & 9.5$\pm$1.0 \\
Suff.\ Context (Gemini) & 40.0 & 24.8 & 3.3 & 62.7 & 48.4 & 30.3$\pm$1.8 & 27.3 & 18.2 & 4.1 & 56.3 & 48.9 & 11.5$\pm$1.2 \\
Abstain-R1 & 51.8 & 51.8 & 6.1 & 50.0 & 27.2 & 52.0$\pm$1.9 & 20.0 & 20.0 & 2.7 & 50.0 & 41.4 & 19.2$\pm$1.2 \\
\midrule\multicolumn{13}{l}{\emph{Decision rules on multi-step evidence, answer from it}} \\
LLM judge, \method{}'s pool & 62.3 & 39.8 & 5.4 & 79.8 & 52.9 & 25.8$\pm$1.9 & 40.9 & 27.2 & 4.7 & 61.6 & 49.2 & 11.2$\pm$1.2 \\
LLM judge, \method{}'s links & 48.0 & 14.1 & 1.9 & \textbf{82.7} & 75.8 & 2.9$\pm$0.7 & 24.1 & 7.9 & 2.1 & 60.7 & \underline{61.0} & $-$0.6$\pm$0.4 \\
Qwen3-8B label, \method{}'s pool & 33.0 & 6.3 & \underline{1.2} & 69.9 & 63.5 & 15.2$\pm$1.6 & 21.0 & 9.8 & 2.3 & 57.0 & 53.6 & 6.8$\pm$1.0 \\
Gemini label, answer from \method{}'s pool & 37.1 & \underline{2.3} & \underline{1.2} & 77.7 & 69.7 & 9.0$\pm$1.4 & 17.9 & 5.4 & 2.3 & 57.6 & 56.4 & 4.1$\pm$0.9 \\
Gemini label, agent's evidence & 28.9 & 4.0 & 1.9 & 67.5 & 67.3 & 11.4$\pm$1.5 & 18.9 & 6.4 & 2.0 & 57.7 & 57.0 & 3.5$\pm$1.1 \\
S2G-RAG, Gemini judge & 31.3 & 5.2 & 1.4 & 69.0 & 69.2 & 9.5$\pm$1.4 & 13.3 & \underline{2.4} & \underline{0.6} & 56.3 & 57.6 & 2.8$\pm$1.0 \\
Qwen3-8B label, Search-R1's evidence & 38.3 & 10.3 & 4.0 & 72.7 & 71.0 & 7.7$\pm$1.6 & 28.9 & 13.2 & 4.1 & 61.1 & 57.6 & 2.9$\pm$1.1 \\
Gemini label, Search-R1's evidence & 40.5 & 5.2 & 2.6 & 79.7 & \underline{78.0} & 0.7$\pm$1.3 & 27.9 & 8.8 & 3.8 & 63.3 & \textbf{62.1} & $-$1.7$\pm$1.0 \\
\midrule
\ours{}\method{} & 44.3 & 8.2 & 1.6 & \underline{82.4} & \textbf{78.7} &  & 22.8 & 7.6 & 2.1 & 59.8 & 60.4 &  \\
\bottomrule
\end{tabular}
\end{table*}

\begin{table*}[t]
\caption{Decision accuracy (\%) in further settings (test pairs, DeepSeek only on those not answerable from memory, Web only on verified pairs). \method{} is best in 6 of the 13 columns, and no row is significantly above it.}
\label{tab:summary}
\centering\footnotesize
\setlength{\tabcolsep}{2.4pt}
\begin{tabularx}{0.8598\textwidth}{l*{13}{>{\raggedleft\arraybackslash}X}}
\toprule
Dec$\uparrow$ & 8B & \multicolumn{3}{c}{Qwen3-32B} & \multicolumn{3}{c}{DeepSeek V4.1 Flash} & \multicolumn{3}{c}{Web, 8B} & \multicolumn{3}{c}{Web, 32B} \\
 & HQ & 2W & MQ & HQ & 2W & MQ & HQ & 2W & MQ & HQ & 2W & MQ & HQ \\ \midrule
RAG (top-5) & 53.0 & 53.0 & 49.8 & 58.5 & 54.5 & 53.1 & 63.1 & 41.6 & 44.0 & 43.2 & 46.0 & 45.1 & 45.2 \\
Iterative search & 38.0 & 57.0 & 46.6 & 53.9 & 73.7 & 63.3 & 70.4 & 23.5 & 22.0 & 28.3 & 39.8 & 27.7 & 39.8 \\
S2G-RAG & 62.5 & 61.5 & 55.2 & 68.9 & 74.3 & 59.9 & 75.3 & 50.4 & 47.4 & 48.4 & 55.8 & 49.8 & 58.4 \\
Search-R1 & 40.1 & 38.9 & 25.7 & 37.4 & 38.5 & 30.8 & 38.0 & 30.1 & 21.2 & 30.9 & 34.5 & 16.2 & 30.9 \\
Suff.\ Context (Gemini) & 61.3 & 54.8 & 51.0 & 64.5 & 54.5 & 53.1 & 65.2 & 42.5 & 45.6 & 48.0 & 49.6 & 46.1 & 48.8 \\
Gemini label, top-5 & 64.2 & 61.2 & 54.4 & 66.5 & 60.0 & 53.9 & 64.6 & 49.6 & \textbf{50.3} & 50.4 & 52.7 & 49.8 & 50.8 \\
LLM judge, \method{}'s links & 71.2 & 77.6 & \underline{64.1} & 73.5 & \underline{83.1} & \underline{65.1} & \textbf{78.0} & 58.8 & 48.4 & 55.2 & 58.4 & 47.6 & 56.2 \\
Gemini label, \method{}'s pool & \underline{71.9} & 75.8 & 62.8 & \textbf{77.1} & 78.2 & 60.3 & 77.4 & 59.7 & 49.7 & 55.2 & \underline{65.5} & \underline{51.1} & \underline{60.0} \\
Qwen3-8B label, Search-R1's evidence & 69.1 & 72.8 & 57.8 & 71.5 & 70.7 & 56.5 & 68.9 & 59.3 & 47.4 & 54.2 & 60.2 & 46.0 & 53.8 \\
Gemini label, Search-R1's evidence & \textbf{73.8} & \textbf{81.4} & 62.7 & \underline{75.6} & 78.4 & 61.2 & 72.8 & \textbf{64.2} & 48.5 & \underline{58.2} & \textbf{69.5} & 48.5 & 56.6 \\
\midrule
\ours{}\method{} & 71.6 & \underline{81.3} & \textbf{64.2} & 75.4 & \textbf{84.4} & \textbf{65.4} & \underline{77.9} & \underline{63.7} & \underline{50.0} & \textbf{58.6} & \underline{65.5} & \textbf{51.5} & \textbf{60.6} \\
\bottomrule
\end{tabularx}
\end{table*}

\subsection{RQ2: Answering or abstaining}
\label{sec:rq2}
Table~\ref{tab:selective} pairs every answerable question that passed the memory filter with its unanswerable twin.
The entity check cuts the twins \method{} answers from \nTWLenientTwin\% and \nMULenientTwin\% (no check, rival values allowed) to \nTWAnvTwin\% and \nMUAnvTwin\%, and decision accuracy rises by \nVTWLenientDDec{} and \nVMULenientDDec{} points, and significantly by \nVPooledLenientDDec{} pooled over the three corpora.
No baseline answers more answerable questions correctly while answering fewer twins than \method{}. Each decides significantly less accurately on every corpus and pooled over the three (S2G-RAG by \nVPooledStwogDDec{} points pooled). Of the \nTableTwoPub{} baselines, \nBelowHalfTW{} score below always abstaining (50) on both 2WikiMultiHopQA and MuSiQue, so the controls below, which abstain too, are the closer comparisons. Abstain-R1, trained to abstain, scores \nAbPooledClosedDec\% pooled over the three corpora, closed-book, and stays below \method{} on \method{}'s own pool on every corpus (\nAbPooledPoolD{} pooled).
With Gemini-2.5-Flash as its judge, S2G-RAG is below \method{} on 2WikiMultiHopQA (\nVTWStwogGemDDec) and pooled (\nVPooledStwogGemDDec), and level on MuSiQue and HotpotQA.
These margins assume half of the queries are unanswerable. \method{} beats Search-R1 in its default configuration once at least \nPrevVTWSr\%, \nPrevVMUSr\% and \nPrevVHOSr\% are. It stays ahead of the Gemini label on its pool unless \nPrevGpMin{} to \nPrevGpMax\% are.

\paragraph{What the evidence and the closure each contribute}
Three controls separate the evidence from the decision rule. The same LLM decides at \nTWSameDec{} and \nMUSameDec\% given \method{}'s pool and requirement list and at \nTWLinksAnvDec{} and \nMULinksAnvDec\% given its accepted table of the values linked to each requirement (decompose-and-verify with an LLM decision~\cite{press2023selfask,khot2023decomposed,xu2024searchchain}).
The closure, deciding that table without a model call, is above it on 2WikiMultiHopQA (\nVTWLinksAnvDDec) and level on MuSiQue and HotpotQA (\nVMULinksAnvDDec{} and \nVHOLinksAnvDDec{} points).
The Gemini Sufficient Context label on \method{}'s pool, which abstains on almost every twin, is below \method{} on 2WikiMultiHopQA, MuSiQue and pooled (\nVPooledPipeGemDDec{} pooled) and level on HotpotQA. The label of Qwen3-8B, whose links the closure reads, is below \method{} by \nVTWPipeSmallDDec{} and \nVMUPipeSmallDDec{} points. Without a model call, the closure decides as accurately as the Gemini label on the same whole-pool answers (level, \nSaPooledDDec{} pooled, \nSaRPooledDDec{} with 32B and \nSaDSPooledDDec{} with DeepSeek). The Gemini margin thus comes from the answer. Conversely, with the label deciding and the certificate answering, decision accuracy is level with \method{}'s (\nHyPooledD{} pooled with 8B).
The margin over the baselines therefore comes from the closure's certified context, the search directed by the requirement hypergraph, and per-requirement linking (RQ3).
The answering model matters as much. Search-R1's answers, kept only when the Gemini label accepts the passages it retrieved, are level with \method{} on every corpus (\nVPooledSrLabDDec{} pooled, and above it pooled with one of two other example sets, \nDvTwoSrLabD). Under the Qwen3-8B label they are below it pooled (\nVPooledSrLabSmallDDec), as is Search-R1's 32B checkpoint under the Qwen3-32B label (\nQtwoD).
With Qwen3-32B in every role (Table~\ref{tab:summary}) every baseline, Search-R1's 32B checkpoint included, stays significantly below \method{} on every corpus and pooled, and that checkpoint under the Gemini label is level (\nVRPooledSrBigLabDDec{} pooled). With DeepSeek, \method{} is significantly above every baseline and Search-R1 under the Gemini label (\nVDSPooledSrLabDDec{} pooled).

\paragraph{Bounding false certificates}
With the entity check \method{} answers \nTWAnvTwin\% and \nMUAnvTwin\% of the twins (\nTWCertTwin\% and \nMUCertTwin\% with the violation rule), and \nTWAnvTwinGold{} and \nMUAnvTwinGold{} points carry the original gold answer.
These are bounded by surviving statements plus accepted false links. A BM25 and full-corpus string search flags \nLeakTW\% and \nLeakMU\% of the twins. Auditing all flagged twins and 100 unflagged ones, Gemini-2.5-Flash confirms a surviving statement for an estimated \nAudTW\% and \nAudMU\%. On pairs whose twin has no confirmed statement, every baseline stays significantly below \method{} (pooled \nCaPubMin{} to \nCaPubMax). Without the check, most false certificates come from links to related-entity pages, a failure of exact linking (H3) that the check targets, while the closure stays exact. With it, a false link, not a surviving statement, explains most answered twins. Gemini grades a link incorrect in \nFlTWTwin, \nFlMUTwin{} and \nFlHOTwin\% of them and in \nFlTWOrig, \nFlMUOrig{} and \nFlHOOrig\% of accepted originals (App.~\ref{app:anchor}).
\emph{Answer:} the certificate decides more accurately than every baseline, as accurately as a Gemini-2.5-Flash label on the same answers, and more accurately than Search-R1 under a label from its own model.

\subsection{RQ3: What the certificate delivers}
\label{sec:cert}
\paragraph{Context and abstention}
The reader's gain in Table~\ref{tab:certificate} is significant with every model. Its passages-only row uses the \nNdN{} pooled answers without a rival value (\nNdB\% from the whole pool). The named requirement is an open requirement whose references are resolved, or the failing lookup. It is counted on traceable twins, the \nLocShare\% of twin abstentions whose original linked the deleted passages and whose plan has two or more lookups, one intact.

\begin{table}[t]
\caption{What the certificate delivers, against the alternative on the same items (test, 8B unless stated, better value in bold). Reader: right on \method{}'s answered originals. Complex: three or more lookups or a join.}
\label{tab:certificate}
\centering\footnotesize
\setlength{\tabcolsep}{2.5pt}
\begin{tabular}{lrr}
\toprule
 & Certificate & Alternative \\ \midrule
\multicolumn{3}{l}{\emph{Context and abstention}} \\
\quad Reader, 8B / 32B / DeepSeek (\%) $\uparrow$ & \textbf{\nAsPooledCert} / \textbf{\nAsRPooledCert} / \textbf{\nAsDSPooledCert} & \nAsPooledPool{} / \nAsRPooledPool{} / \nAsDSPooledPool{} (pool) \\
\quad Reader, passages only (\%) $\uparrow$ & \textbf{\nNdA} & \nNdC{} (RAG, same size) \\
\quad Names the deleted requirement (\%) $\uparrow$ & \textbf{\nLocHit} & \nLocChance{} (at random) \\
\multicolumn{3}{l}{\emph{Decision on the same links, with no model call}} \\
\quad Originals right, only it answers $\uparrow$ & \textbf{\nSwClosOk{} of \nSwClosN} & \nSwOrigOk{} of \nSwOrigRest{} (step-wise) \\
\quad Twins answered, only it answers $\downarrow$ & \textbf{\nSwClosTwin} & \nSwTwinAns{} (step-wise) \\
\multicolumn{3}{l}{\emph{Complex questions, \nHxShare\% of pairs (\%)}} \\
\quad Originals right when answered $\uparrow$ & \textbf{\nHxCertSAcc} & \nHxSrSAcc{} (Search-R1) \\
\quad Twins answered $\downarrow$ & \textbf{\nHxCertTwin} & \nHxSrTwin{} (Search-R1) \\
\bottomrule
\end{tabular}
\end{table}

\paragraph{The decision}
With decomposition, search and accepted links fixed, the closure decides without a model call. On chains of lookups it is the step-wise rule ``every lookup has an accepted link''~\cite{trivedi2022musique}, and they differ on \nSwDisMin{} to \nSwDisMax\% of all decisions. Beyond chains it answers only when the target itself is derived. Where only the step-wise rule answers, a reasoning requirement has no value, and where only the closure answers, a lookup its derivation does not need is open (Table~\ref{tab:certificate}). The closure leads an LLM on the same links most on 2WikiMultiHopQA kinship chains (\nBTTWInfJudge{} on \nBTTWInfPairs{} pairs).

\paragraph{Complex questions}
On the \nHxN{} complex pairs, \method{} stays significantly above every baseline (by \nHxMin{} to \nHxMax{} points), carried by twin abstentions and accurate answers.

\paragraph{Cost}
Certification costs mostly linking calls, each judging one passage for one requirement (\nTWAnvCalls{} and \nMUAnvCalls{} LLM calls per question against Search-R1's \nTWSrCalls, verification under one). Only HALT and SR-RAG use more calls or tokens.
\emph{Answer:} without a model call, the certificate decides, makes the reader more accurate and mostly names the deleted requirement.

\subsection{RQ4: The hypotheses in practice}
\label{sec:hyp}
\begin{table}[t]
\caption{Precision of \method{}'s positive links by acceptance class (test, 8B, up to 150 links per class and all \nAudVerTWN{} verified 2WikiMultiHopQA links graded by Gemini-2.5-Flash, correct/graded in parentheses). Anchored links are the most precise in every row (bold). N/A: the Web rule has no verification step.}
\label{tab:audit}
\centering\footnotesize
\setlength{\tabcolsep}{1.6pt}
\begin{tabular}{llrrrr}
\toprule
 & & Accepted & \multicolumn{3}{c}{Precision (\%)$\uparrow$} \\
 & & share (\%) & anchored & verified & rejected \\ \midrule
Pooled & 2WikiMultiHopQA & 77 & \textbf{96.7 (145/150)} & 43.2 (38/88) & 20.0 (30/150) \\
 & MuSiQue & 50 & \textbf{84.7 (127/150)} & 74.0 (111/150) & 54.7 (82/150) \\
 & HotpotQA & 61 & \textbf{82.7 (124/150)} & 74.7 (112/150) & 47.3 (71/150) \\
\midrule
Web & 2WikiMultiHopQA & 56 & \textbf{76.7 (115/150)} & N/A & 44.7 (67/150) \\
 & MuSiQue & 44 & \textbf{79.3 (119/150)} & N/A & 54.0 (81/150) \\
 & HotpotQA & 60 & \textbf{74.0 (111/150)} & N/A & 50.7 (76/150) \\
\bottomrule
\end{tabular}
\end{table}

\looseness=-1 Gold retrieval (H1) holds for \nCauHoneTW\% and \nCauHoneMU\% of the questions, and the accepted certificate fires for \nCauFireTW\% and \nCauFireMU\% of these.
In an audit of up to 150 links per class, anchored links are far more often correct than rejected ones on every corpus and at both scales (Table~\ref{tab:audit}, \nAudVAncTW\% against \nAudVUnTW\% on 2WikiMultiHopQA).
A blind annotator (an author) confirms this on \nHuLinks{} links (anchored \nHuAnc\%, rejected \nHuOth\% correct) and agrees with Gemini-2.5-Flash on \nHuGem\% ($\kappa=\nHuGemK$).
\emph{Answer:} gold retrieval holds for \nCauHoneTW\% and \nCauHoneMU\% of questions, and exact linking far more often on anchored than on rejected links.

\subsection{RQ5: Exact requirements}
\label{sec:rq4}
On 959 TextWorld episodes, the certificate recovers every enabling set and enables no twin. Pairwise rules fail on the test episodes, as directed blindness predicts. A prover with every pairwise link but no conjunction recovers 58.6\% at depth two or more, and reachability in timed and untimed projections enables 17.2\% and 84.2\% of twins. An LLM on the same tables falls from \nELBigOneDec\% at depth one to \nELBigDeepDec\% at depth ${\ge}4$ (32B). In ProofWriter several rules can yield a fact. A step-wise plan fixing one rule per requirement before reading the facts misses \nPwMiss\% of provable questions (\nPwMissDeep\% at depth 5, \nPwRandMin{} to \nPwRandMax\% with random rule orders), and the closure none. \emph{Answer:} the closure recovers every enabling set and provable answer that pairwise views, committed plans and an LLM on the same tables miss.

\subsection{Web-scale retrieval: anchoring the links}
\label{sec:web}
\looseness=-1 We replace the pooled corpora by the 21M-passage 2018 Wikipedia dump of DPR~\cite{karpukhin2020dpr}.
A twin deletes every passage among the top 50 hits that states the deleted fact, leaving \nWTWPairs, \nWMUPairs{} and \nWHOTwins{} test pairs.
As the 2018 snapshot lacks the gold support of many questions, we report the \nWVTWPairs, \nWVMUPairs{} and \nWHOPairs{} gold-verified pairs with Qwen3-8B and -32B in every role (Table~\ref{tab:summary}). In these, every lookup has a passage titled by its gold page that states its gold value.
Unanchored, \method{} stays near always abstaining (\nWCertSTWDec, \nWCertSMUDec{} and \nWCertSHODec\% with 8B). In most of the twins it answers at 32B over all pairs, the linker accepts a passage about another entity at the deleted hop, often a near-namesake.
Subject anchoring removes most such links and raises decision accuracy in all six cells, by \nWGainSTW, \nWGainSMU{} and \nWGainSHO{} with 8B and by \nWGainBTW, \nWGainBMU{} and \nWGainBHO{} with 32B (Table~\ref{tab:summary}). On all pairs the gain is \nWaGainMin{} to \nWaGainMax{} points.
Pooled, anchored \method{} is significantly above every baseline with 8B and all but S2G-RAG with 32B, and level with Search-R1 under the Gemini label (\nAWSPooledSrLabDDec{} at 8B, \nAWBPooledSrBigLabDDec{} at 32B).
Without the twins that keep a surviving statement, which no linking rule can remove (at least \nWAudTW, \nWAudMU{} and \nWAudHO\% of all twins), it decides \nWCTWAncDec, \nWCMUAncDec{} and \nWCHOAncDec\% (8B).

\section{Conclusion}
\label{sec:conclusion}

We cast the decision to answer, search again or abstain as evidence sufficiency over a per-query requirement hypergraph.
Its Horn closure decides without a model call, returns a minimal sufficient context with its derivation, and names the open requirements that direct the next search.
When retrieval, requirements and links are correct, this decision is sound and complete up to the entity check.
Where requirements are exact, as in agent logs and rule-based theories, the closure recovers every derivation, whereas rules reading only pairwise projections provably miss conjunctions and plans committed in advance miss alternatives.
On three pooled corpora \method{} decides significantly more accurately than search agents, graph retrieval and sufficiency judges in default configurations.
On the same evidence its closure matches a fast proprietary model's sufficiency label, and its certificate makes the reader more accurate.
At Web scale, anchoring links to entity titles raises its decision accuracy on every corpus and model size.\label{body:end}

\bibliographystyle{ACM-Reference-Format}
\bibliography{refs_compact}

\appendix
\section{Proofs}
\label{app:proofs}

\paragraph{Lemma~\ref{lem:closure}.}
If $X\subseteq X'$, $S\subseteq S'$ and $E\subseteq E'$, every head added by $\Phi_{S,E}(X)$ is added by $\Phi_{S',E'}(X')$, because $\Req(r)\subseteq X$ implies $\Req(r)\subseteq X'$ and $r\in E$ implies $r\in E'$.
On the finite lattice $2^{\Atoms_q}$ the least fixpoint of a monotone map exists and equals $\bigcup_{i\ge0}\Phi^i_{S,E}(\emptyset)$ (Knaster--Tarski). Ordered maps have ordered least fixpoints, so $\Cl$ is monotone in $S$ and $E$.
$S\subseteq\Phi_{S,E}(X)$ for every $X$ gives extensivity. $\Cl(S,E)$ being a fixpoint of $\Phi_{\Cl(S,E),E}$ gives idempotence.
\looseness=-1 Sufficiency is monotone in the context as $\facts(\Ctx)$ and $\rules(\Ctx)$ are.

\paragraph{Proposition~\ref{thm:cert}.}
(1) By (H1), $F_q\subseteq\mathit{pool}$, so monotonicity gives
\begin{multline*}\Target_q\subseteq\Cl\big(G_q\cup\facts^*(F_q),\rules^*(F_q)\big)\\ \subseteq\Cl\big(G_q\cup\facts(\mathit{pool}),\rules(\mathit{pool})\big),\end{multline*}
using (H2) and (H3) to replace the true maps by the computed ones on the pool. So the closure test succeeds, and only a violation or a lookup link of the derivation that fails the entity check (Algorithm~\ref{alg:certify}) can lead to abstention.
(2) \method{} answers only if $\Ctx\subseteq\mathit{pool}$ passes the sufficiency test. Under (H2)--(H3) the test computes the true closure, and as each lookup requirement has one true value no rival is linked, so no hyperedge of the derivation is violated under either rule.
(3) Let $\Ctx_0$ be the backward-derivation set, sufficient by construction, and $u_1,\dots,u_m$ the pruning order. $\Ctx_i=\Ctx_{i-1}\setminus\{u_i\}$ if that set is sufficient, else $\Ctx_i=\Ctx_{i-1}$, so every $\Ctx_i$ is sufficient. If $u_i\in\Ctx_m$ it was kept, so $\Ctx_{i-1}\setminus\{u_i\}$ and, by monotonicity, its subset $\Ctx_m\setminus\{u_i\}$ are insufficient. Every proper subset of $\Ctx_m$ lies in some $\Ctx_m\setminus\{u\}$.
(4) Under the single-derivation condition, every sufficient subset of the pool derives $\Target_q$ by that derivation, since a hyperedge that fires in a subset fires in the pool (Lemma~\ref{lem:closure}) and no other fires there. So it contains the set $P$ of the unique producers of its stated atoms and witnessed hyperedges, and $P$ is sufficient, hence the unique minimal sufficient subset, which pruning returns by (3).
(5) The forward closure, which yields $\Ctx_0$, is one computation and each pruning test one more. \hfill$\square$

\paragraph{Proposition~\ref{prop:directed}.}
$\Proj(\tau\Rules)=\Proj(\Rules)$ and $\tau$ changes nothing else that $D$ reads, so $D(\mathcal M)=D(\tau\mathcal M)$.
For a pair whose true labels differ, $D$ is wrong on exactly one member. Under the balanced mixture each member has probability one half, so the error is at least half the probability that the labels differ. \hfill$\square$

\paragraph{Proposition~\ref{prop:undirected}.}
$\Sec(\mathbf B)$ is the union over units of the pairs inside each unit. Every pair inside $u^\star$ lies inside another unit, so removing $u^\star$ removes no pair. A rule reading only the query, entity contents and $\Sec(\mathbf B)$ thus sees identical inputs on both collections. If $u^\star$ is the only unit stating an atom as in the proposition, removing it makes the targets underivable, so the true labels differ, and the two-point argument of Proposition~\ref{prop:directed} applies. \hfill$\square$

\section{Implementation Details}
\label{app:impl}

\paragraph{Search, linking and decomposition}
\looseness=-1 For each lookup requirement, \method{} retrieves up to three rounds of 5, 5 and 10 new passages, queried by the instantiated requirement, until one passage states an answer.
The linker sees one passage and one instantiated requirement and returns a short answer or \texttt{NONE}, accepting explicit or directly implied answers and rejecting passages about another entity.
One reader call with four in-context examples builds the hypergraph. Two other example sets change \method{}'s decision accuracy by \nDvOneShift{} and \nDvTwoShift{} pooled. Every baseline, the Gemini label on its pool and Search-R1 under the Qwen3-8B label stay significantly below it.

\paragraph{Baselines}
\looseness=-1 RAG tells the reader to answer \texttt{unanswerable} if the documents are insufficient. Sufficient Context uses its published one-shot prompt. The label and the reader's $P(\mathrm{True})$ enter a logistic head fitted on validation to the coverage of the unanchored certificate with the violation rule. The iterative agent alternates search queries and a final answer, may answer \texttt{unanswerable}, and stops before its limit of four searches on \nIterStopMin{} to \nIterStopMax\% of questions. With one search round of five passages per requirement, \method{} links \nBudTWLinks, \nBudMULinks{} and \nBudHOLinks{} passages per question (default \nBudTWLinksDef, \nBudMULinksDef{} and \nBudHOLinksDef) and decides as accurately (\nBudTWDec, \nBudMUDec{} and \nBudHODec\%, fresh run). With DeepSeek V4.1 Flash as judge (agreement \nDjMin--\nDjMax\% for other methods with 50+ answers), every baseline stays significantly below it (pooled from \nDjPubMin). Under it, S2G-RAG with the Gemini judge and Search-R1 under the Qwen3-8B label stay below (\nDjStwogGemD, \nDjSrSmD) and Search-R1 under the Gemini label level (\nDjSrLabD).
HALT uses the authors' claim prompt and released LoRA verifier. Answering at budget end by default, HALT and S2G-RAG fall further behind (HALT pooled \nVPooledHaltADDec; S2G-RAG \nVTWStwogAllDDec, \nVMUStwogAllDDec).
S2G-RAG runs four six-passage turns with its evidence selector (its distilled judge is not public).
Search-R1 runs the authors' Qwen2.5-7B PPO checkpoint and, with Qwen3-32B, their only released 32B one (Qwen2.5-32B GRPO). Both are trained on Natural Questions and HotpotQA training questions (ours: validation) and run with their settings and our retriever. Without the \nSrMemTWKnown{}, \nSrMemMUKnown{} and \nSrMemHOKnown{} paired originals its model answers from memory, \method{} still leads by \nSrMemTWD{}, \nSrMemMUD{} and \nSrMemHOD{} points.
Abstain-R1 runs the authors' checkpoint and prompt (greedy, at most 4{,}096 tokens). A boxed ``I don't know'' or none abstains. On \method{}'s pool the passages precede the question, an input it was not trained on, and where its closed-book reading misses the original, \method{} still leads by \nAbMemPooledPoolD{} (pool) and \nAbMemPooledRagD{} (top-5). HippoRAG~2 runs the authors' code (personalised PageRank over a Qwen3-8B entity graph) without a twin's deleted passages.
LogicRAG runs the authors' code and prompts ($k=3$, at most five rounds). Its answer step reads an LLM-written summary that can carry facts from memory.
SR-RAG (no public code) runs its paper's prompts and settings (up to three sub-queries of three passages, one feedback round), with our few-shot examples, fallback prompt and graph update. Counting its answers that say the documents lack the answer as abstentions, \method{} still leads it by \nSrgRefPooledD{} pooled.

\paragraph{Unanswerable twins}
\looseness=-1 Of 1{,}000 questions each, 24 (2WikiMultiHopQA), 54 (MuSiQue) and 239 (HotpotQA) yield no valid twin, and 447, 120 and 274 of the rest are answered from memory. The deleted 2WikiMultiHopQA and HotpotQA hop is drawn per question with a fixed seed (with a second seed no method changes rank), and every bootstrap uses a fixed seed.
The audit asks Gemini-2.5-Flash (DeepSeek on HotpotQA, whose hops Gemini wrote) whether flagged or dense top-10 passages state the deleted fact.
Deleting 1--4 passages leaves every reader context the same size, and a length-only classifier tells originals from twins at chance (AUROC 0.50--0.52)~\cite{mondal2026before}.

\paragraph{HotpotQA and Web scale}
On HotpotQA, HippoRAG~2, LogicRAG, SR-RAG, HALT, Verbalised sufficiency and Abstain-R1 decide significantly less accurately than \method{}, by \nHqRestMin{} to \nHqRestMax{} points. On all Web pairs, anchoring raises \method{} by \nWaGainMin{} to \nWaGainMax{} points, numerically at or above every baseline.
\looseness=-1 Each HotpotQA lookup by Gemini-2.5-Flash is tied to a supporting passage. A question is kept only if each hop's answer occurs in its passage and a bridge's second lookup uses the first and yields the gold answer (779 of 1{,}000; \nHOPairs{} test pairs). These lookups define the twins and the oracle decomposition.
Web twins delete \nWTWDeleted, \nWMUDeleted{} and \nWHODeleted{} passages on average (facts stated 9+ times dropped). Search-R1's 32B Web rows use a same-settings rebuilt index (recall@10 \nWfRecall{} vs.\ exact, first \nWfRecallOrig).

\section{Entity Anchoring and Verification}
\label{app:anchor}

\paragraph{Definition}
\looseness=-1 Write $W(x)$ for the set of lower-cased word tokens (maximal runs of letters and digits) of a string $x$ that are longer than one character and not among \emph{of, de, the, bin, von, van, la, le, del, da, di, and, du, al, el, y}. Write $t(u)$ for the title of unit $u$ with its disambiguating parenthetical removed.
A unit $u$ linked to the instantiated lookup requirement $q_k$ (of $s_k$) with value $v$ is \emph{subject-anchored} if $\emptyset\neq W(t(u))\subseteq W(q_k)$, and \emph{value-anchored} if $W(t(u))\neq\emptyset$ and $W(t(u))\subseteq W(v)$ or $\emptyset\neq W(v)\subseteq W(t(u))$.
A lookup requirement $k$ of the derivation whose chosen atom $a_k(v_k)$ has no subject-anchored supporting unit is \emph{verified} if the reader, given $q_k$, $v_k$ and the title and text of one of the first two units supporting $a_k(v_k)$ (in linking order), answers Yes to whether the passage states that $v_k$ answers $q_k$ for exactly the entity $q_k$ names, which accepts that unit's link only.
The accepted certificate is Algorithm~\ref{alg:certify}'s certificate with rival values allowed, accepted only if every lookup requirement $k$ of the derivation is anchored or, on the pooled corpora, verified. Otherwise \method{} abstains and names requirement $k$. Its context uses accepted links only: a unit stating facts of two steps but accepted for one supports only that step.
\looseness=-1 As later hops are instantiated with earlier values, a subject anchor at a hop with $\mathrm{ref}(k)\neq\emptyset$ checks the earlier bridge entity, so anchoring covers the derivation. At Web scale verification admits near-namesakes (43 of 65 twin steps).

\paragraph{Link audit}
\looseness=-1 The false-link rates of Section~\ref{sec:rq2} grade every supporting link of 150 accepted originals per corpus and of all answered twins. Of the answered twins that the surviving-statement audit covers, all but \nFlTWNeither, \nFlMUNeither{} and \nFlHONeither{} (of \nFlTWAud, \nFlMUAud{} and \nFlHOAud) have a false link or a surviving statement. Grading the same links, DeepSeek V4.1 Flash keeps anchored links more precise than rejected ones on every corpus and finds fewer false derivations.

\paragraph{Human audit}
\looseness=-1 The annotator labelled \nHuLinks{} of the 270 links drawn with a fixed seed and \nHuTwins{} twins (25 flagged and 25 unflagged per corpus), blind to the class and the model verdicts. Verified links are correct on \nHuVer\%. DeepSeek agrees on \nHuDs\% ($\kappa=\nHuDsK$). Gemini mostly differs by rejecting links the annotator accepts (\nHuGemNo{} against \nHuGemYes). Over all twins, the annotator's estimated surviving-statement rates are \nHuLeakTW{} and \nHuLeakMU\% (Gemini \nAudTW{} and \nAudMU\%).

\paragraph{Further settings}
\looseness=-1 \method{}'s per-item 8B rows, lost with the original machines, were rebuilt from the stored links by re-running the identity checks. At most four decisions and seven answers per corpus changed (Dec \nRbTWDec{} and \nRbMUDec{} on 2WikiMultiHopQA and MuSiQue). They enter only differences and RQ3's step-wise comparison.
Without pairs whose twin keeps a confirmed surviving statement (\nCaTWDropped, \nCaMUDropped, \nCaHODropped), or any flagged one, every baseline stays significantly below \method{} on every corpus (pooled \nCaPubMin{} to \nCaPubMax). So does the Gemini label on its pool (\nCaLab), while Search-R1 under it stays level (\nCaSrLab). On the audited-clean Web twins, anchored \method{} is above every baseline in all six cells (by \nWCPubMin{} to \nWCPubMax).

\paragraph{A complex case}
\looseness=-1 For \emph{Which film has the director born later, Je Fais Le Mort or Red Roses, Red Lips, Red Wine?} (four lookups and a comparison, Figure~\ref{fig:framework}), \method{} certifies Je Fais Le Mort from four passages: its director Jean-Paul Salom\'e was born on 14 September 1960, Paul Martin on 8 February 1899. On the twin it abstains and names ``When was Jean-Paul Salom\'e born?'', while Search-R1, LogicRAG and S2G-RAG answer it and Search-R1 and LogicRAG also miss the original. The case is drawn by a fixed rule from the \nHxCases{} four-lookup pairs that \method{} decides correctly and two or more agents miss.

\clearpage
\end{document}